\documentclass[journal]{IEEEtran}
\usepackage{amsmath,amssymb,amsthm,amsfonts}
\theoremstyle{definition}
\newtheorem{thm}{Theorem}

\newtheorem{lem}{Lemma}
\newtheorem{cor}{Corollary}

\newtheorem{prop}{Proposition}
\usepackage{bm}
\usepackage[textsize=footnotesize,textwidth=2cm]{todonotes}  %,disable
\usepackage{tikz}
\usepackage{verbatim}
\usepackage{gensymb}
\usepackage{graphicx}
\usepackage{subcaption}
\usepackage{multirow}
\usepackage{etoolbox}
\usepackage{xcolor}
\usepackage{authblk,algpseudocode,algorithm,fancyvrb,enumitem}

\algrenewcommand\textproc{\texttt}

\usepackage[hidelinks]{hyperref}

\usepackage[capitalise]{cleveref}
\crefname{thm}{Theorem}{Theorems}
\crefname{defn}{Definition}{Definitions}
\crefname{lem}{Lemma}{Lemmas} % or 'lemmata'?
\crefname{cor}{Corollary}{Corollaries}
\crefname{conj}{Conjecture}{Conjectures}
\crefname{prop}{Proposition}{Propositions}
\crefname{tbl}{Table}{Tables}
\AtBeginEnvironment{appendices}{\crefalias{section}{appendix}}

\usepackage{orcidlink}

\usepackage{fontawesome}

\begin{document}

\title{An insightful approach to bearings-only tracking in log-polar coordinates}

\author{Athena Helena Xiourouppa\,\orcidlink{0000-0003-0033-744X}, Dmitry Mikhin\,\orcidlink{0009-0002-6228-2958}
Melissa Humphries\,\orcidlink{0000-0002-0473-7611}, and John Maclean\,\orcidlink{0000-0002-5533-0838}}
        % <-this % stops a space



\maketitle

\begin{abstract}
    The choice of coordinate system in a bearings-only (BO) tracking problem influences the methods used to observe and
    predict the state of a moving target. Modified Polar Coordinates (MPC) and Log-Polar Coordinates (LPC) have some
    advantages over Cartesian coordinates. In this paper, we derive closed-form expressions for the target state prior
    distribution after ownship manoeuvre: the mean, covariance, and higher-order moments in LPC. We explore the use of
    these closed-form expressions in simulation by modifying an existing BO tracker that uses the UKF. Rather
    than propagating sigma points, we directly substitute current values of the mean and covariance into the time update
    equations at the ownship turn. This modified UKF, the CFE-UKF, performs similarly to the pure UKF, verifying the
    closed-form expressions. The closed-form third and fourth central moments indicate non-Gaussianity of the target
    state when the ownship turns. By monitoring these metrics and appropriately initialising relative range error, we
    can achieve a desired output mean estimated range error (MRE). The availability of these higher-order moments
    facilitates other extensions of the tracker not possible with a standard UKF.
\end{abstract}


\section{Introduction} \label{sec:intro}
Consider a sensing platform (ownship) tracking a vessel of interest (target) using a passive acoustic
sensor.\footnote{For simplicity, we ignore the potential differences between the ownship coordinates / orientation and
the sensor coordinates / orientation.} Assume energy propagates in simple, horizontal straight lines, so that there is
only one energy path (ray) from target to sensor. The sensor obtains the target bearing by measuring the direction of
target noise arrival. We aim to recover the target's full state, \textit{i.e.,} position and speed, from a sequence of
this data.

Modified Polar Coordinates (MPC) have been established as a stable approach for bearings-only (BO) target tracking. State
update equations in MPC are non-linear; to implement estimation for non-linear systems, existing literature suggests the
use of the Extended Kalman Filter (EKF)~\cite{Aidala:1983:Utilization,Hoelzer:1978:Modified} or Unscented Kalman Filter
(UKF)~\cite{Wang:2010:Algorithm}.

\textit{Log}-Polar Coordinates (LPC) are another viable candidate \cite{Brehard:2006:Closed, Mallick:2011:Angle},
especially for theoretical analysis; using LPC, one can obtain the posterior Cramér-Rao bound (PCRB) for the target
state distribution \cite{Brehard:2006:Closed}. This paper further explores the theoretical properties of the state
distribution in LPC.

We do so in the context of a simple explicit ownship manoeuvre, described in~\cref{sec:Prelims}; broader scenarios can
be approximated as combinations of simple manoeuvres.

The main contribution of this paper is proving that the target state distribution after an ownship manoeuvre can be
described in LPC with closed-form expressions, derived in~\cref{sec:Distribution}. Closed-form expressions allow us to
better understand the relationships between target state variables. We demonstrate the advantage of this for tracking
simulation in~\cref{sec:Applications} by implementing a Kalman Filter that incorporates the closed-form expressions,
CFE-UKF, and comparing its capabilities to a pure UKF. Conclusions are presented in \cref{sec:Conclude}.

\section{Passive Bearings-Only Tracking\label{sec:Prelims}}
% We begin by introducing notation and assumptions for the target tracking problem in MPC in \cref{sec:Straight}.
% We then discuss the validity of using LPC in \cref{sec:MPCvsLPC} and further simplify the problem to refine the target
% state description in \cref{sec:ExplicitManoeuvre}.

\subsection{Straight Line Movement\label{sec:Straight}}
The target's state in MPC~\cite{Hoelzer:1978:Modified} is described as a vector of
\begin{itemize}
    \item absolute bearing $\beta$ that is the angle from the receiver to the target measured from True North,
	\item bearing rate $\dot{\beta}$,
	\item scaled range rate $\dot{\rho} = \dot{r} / r$, where $r$ is range, and
	\item inverse range $s = 1 / r$.
\end{itemize}
Here and below, the upper dot denotes the derivative with respect to time.

\IEEEpubidadjcol

We define the target state at time $t$ as $\bm{x} = [\beta, \dot{\beta}, \dot{\rho}, s]$. If we assume that the ownship
and target do not manoeuvre, the vessel trajectory in MPC is:
\begin{equation}
    \label{eq:MPCTransition}
    \begin{aligned}
    \beta^+ &= \beta + \arctan{\frac{t_{\beta}}{1 + t_{\rho}}}, &
        s^+ &= \frac{s}{ \sqrt{ {(1 + t_{\rho})}^2 + {t_{\beta}}^2 } },\\
        \dot{\beta}^+ &= \frac{\dot{\beta}}{ {(1 + t_{\rho})}^2 + {t_{\beta}}^2 }, &
        \dot{\rho}^+ &= \frac{\dot{\rho} + \Delta t (\dot{\rho}^2 + \dot{\beta}^2) }{ {(1 + t_{\rho})}^2 + {t_{\beta}}^2 }. \\
    \end{aligned}
\end{equation}
Here $t^+$ is some later time, $\Delta t = t^+ - t$, $t_{\beta} = \dot{\beta} \Delta t$ and $t_{\rho} =
\dot{\rho} \Delta t$. The update equations in~\cref{eq:MPCTransition} are non-linear. To use these equations for target state
updates in a sequential estimator, we must apply the EKF, UKF, or other advanced methods of non-linear estimation
(\textit{e.g.},~\cite{Haykin:2001:Kalman}). By taking $\Delta t$ as the update interval of the EKF or UKF,
\cref{eq:MPCTransition} would define the state transition function for the filter.

The bearing transition equation $\beta^+$ in~\cref{eq:MPCTransition} does not contain the initial value of inverse range
$s$. Hence, we cannot determine the range from BO measurements when the ownship and target move straight.

\subsection{Equations of Motion for Manoeuvring Platforms\label{sec:EqMotion}}
Assume that at the initial time $t$ the ownship position and speed in a fixed Cartesian coordinate system are
$x_{\text{O}}$, $y_{\text{O}}$, $v_{x\text{O}}$, and $v_{y\text{O}}$, and the corresponding position and speed at a
later time $t^+$ are $x^+_{\text{O}}$, $y^+_{\text{O}}$, $v^+_{x\text{O}}$, and $v^+_{y\text{O}}$. Define the ownship
displacement vector and speed change vector as
\begin{equation}
\label{eq:CartDisplacement}
\begin{aligned}
\bm{w}_{\text{O}} &= \left[ x^+_{\text{O}} - x_{\text{O}} - v_{x\text{O}} \Delta t, y^+_{\text{O}} - y_{\text{O}} - v_{y\text{O}} \Delta t \right], \\
\dot{\bm{w}}_{\text{O}} &= \left[ v^+_{x\text{O}} - v_{x\text{O}}, v^+_{y\text{O}} - v_{y\text{O}} \right].
\end{aligned}
\end{equation}
Similarly, define $\bm{w}_{\text{T}}$ and $\dot{\bm{w}}_{\text{T}}$ for the target displacement. The overall
acceleration-related displacement and speed change are
\begin{equation}
\label{eq:Accelerated}
\begin{aligned}
\bm{a} &= \bm{w}_{\text{T}} - \bm{w}_{\text{O}}, \\
\dot{\bm{a}} &= \dot{\bm{w}}_{\text{T}} - \dot{\bm{w}}_{\text{O}}.
\end{aligned}
\end{equation}
Using these definitions, target movement in MPC in the presence of target and ownship manoeuvres follows the equations
\begin{equation}
\label{eq:MPCTransitionM}
\begin{aligned}
\beta^+ &= \beta + \arctan{\frac{D_\beta}{D_\rho}}, &
    s^+ &= \frac{s}{ {\left( D^2_\beta + D^2_\rho \right)}^{1/2} }, \\
\dot{\beta}^+ &= \frac{\dot{D}_\beta D_\rho - \dot{D}_\rho D_\beta}{D^2_\beta + D^2_\rho}, &
    \dot{\rho}^+ &= \frac{\dot{D}_\rho D_\rho + \dot{D}_\beta D_\beta}{D^2_\beta + D^2_\rho}, \\
\end{aligned}
\end{equation}
where
\begin{equation}
\label{eq:MPCTransitionMA}
\begin{aligned}
D_\beta &= t_{\beta} + s ( \bm{a} \cdot \bm{n} ), & \dot{D}_\beta &= \dot{\beta} + s ( \dot{\bm{a}} \cdot \bm{n} ), \\
D_\rho &= 1 + t_{\rho} + s ( \bm{a} \cdot \bm{q} ), & \dot{D}_\rho &= \dot{\rho} + s ( \dot{\bm{a}} \cdot \bm{q} ), \\
\end{aligned}
\end{equation}
and the unit vectors $\bm{q}$ and $\bm{n}$ point from the ownship to the target and orthogonally at time $t$:
\begin{equation}
\label{eq:MPCTransitionMB}
\begin{aligned}
\bm{q} &= \left[ \sin{\beta}, \cos{\beta} \right], \\
\bm{n} &= \left[ \cos{\beta}, -\sin{\beta} \right].
\end{aligned}
\end{equation}
The bearing update equation \labelcref{eq:MPCTransitionM} depends on the initial target range $s$. If we know the
\textit{ownship} manoeuvres, we can estimate $s$, assuming that the target moves straight. We often generalise this
conclusion as ``the ownship must outmanoeuvre the target to determine its range''~\cite{Hoelzer:1978:Modified}. For the
non-manoeuvring target $\bm{w}_{\text{T}} = \dot{\bm{w}}_{\text{T}} = 0$, so the subscript ``O'' can be omitted
in~\cref{eq:CartDisplacement,eq:Accelerated} without causing ambiguity. From the Kalman Filter viewpoint, the
acceleration-related terms $(\bm{a} \cdot \bm{n})$, $(\dot{\bm{a}} \cdot \bm{n})$, $(\bm{a} \cdot \bm{q})$, and
$(\dot{\bm{a}} \cdot \bm{q})$ now represent control inputs to the system.

\subsection{Modified Polar Coordinates Versus Log-Polar Coordinates\label{sec:MPCvsLPC}}
BO trackers commonly use bearing, bearing rate, and scaled range rate as state variables because these values are
directly observable on a straight leg of the ownship~\cite{Hoelzer:1978:Modified}. The choice of the fourth state
component is not fixed. One conventional choice is MPC, in which we use inverse range $s$ as the fourth coordinate, as
in~\cref{sec:Straight,sec:EqMotion}. An alternative is to select the logarithm of range $\rho = \ln r$ as the fourth
state component; this gives LPC with the state vector $\bm{x} = [\beta, \dot{\beta}, \dot{\rho}, \rho]$.

The choice of LPC has several theoretical advantages over MPC. Firstly, it allows the derivation of closed-form
expressions for the posterior Cramér-Rao bound (PCRB) in a BO tracking problem~\cite{Brehard:2006:Closed}.
Compared to MPC, the LPC state vector appears more intuitive because it contains two values (bearing and log-range) and
their time derivatives, while MPC mixes $s$ and $\dot{\rho}$. Also, a Gaussian distribution of the fourth state vector
component $\rho$ never results in zero or negative range values, which theoretically is always possible in MPC at the
tails of the distribution. A corollary is that sigma points of the UKF algorithm~\cite{Wan:2000:Unscented} would never
end up at infeasible negative ranges.

The transition from MPC to LPC is simple: we describe the target motion in LPC by the system of differential
equations
\begin{equation}
    \label{eq:MotionLPC}
    \frac{d}{dt}\begin{bmatrix}\beta \\ \dot{\beta} \\ \dot{\rho} \\ \rho\end{bmatrix} =
        \begin{bmatrix}\dot{\beta} \\ -2 \dot{\beta}\dot{\rho} \\ \dot{\beta}^2 - \dot{\rho}^2 \\
        \dot{\rho}\end{bmatrix},
\end{equation}
obtained from Eqs.~(11) and (18) of~\cite{Hoelzer:1978:Modified} by substituting $s = e^{-\rho}$. The same substitution
into~\cref{eq:MPCTransition,eq:MPCTransitionM,eq:MPCTransitionMA} provides the solutions of~\cref{eq:MotionLPC} for
straight and manoeuvring vessels, respectively. For use in the UKF, one should also adapt the range-related elements of
the process noise matrix $Q$.

In accordance with Kalman Filter theory assumptions, we model the target state vector in LPC, $\bm{x} = [\beta,
\dot{\beta}, \dot{\rho}, \rho]$ as a Gaussian random variable:
\begin{equation}
\bm{x} \sim \mathcal{N}(\bm{x}; \bm{\mu}, \Sigma).
\end{equation}

With LPC in mind, we return to the equations of motion and make a key simplification.

\subsection{Instant Manoeuvres\label{sec:ExplicitManoeuvre}}
Assume that the ownship makes an elementary manoeuvre: an instant turn. The state space variables after the manoeuvre,
$[\beta^+, \dot{\beta}^+, \dot{\rho}^+, \rho^+]$, are related to the corresponding values before the turn, $[\beta,
\dot{\beta}, \dot{\rho}, \rho]$ by~\cref{eq:MPCTransitionM,eq:MPCTransitionMA} with $\Delta t = 0$. The target
coordinates and speed do not change, and neither do the ownship coordinates, as the turn is instant. Therefore, the only
non-zero acceleration-related term is $\dot{\bm{a}} = -\dot{\bm{w}} = -\left[ \Delta v_x, \Delta v_y \right]$,
and~\cref{eq:MPCTransitionM} simplifies to
\begin{equation}
\label{eq:MPCTransitionM00}
\begin{aligned}
\beta^+ &= \beta, & \rho^+ &= \rho, \\
    \dot{\beta}^+ &= \dot{\beta} - \frac{( \dot{\bm{w}} \cdot \bm{n} )}{r}, &
    \dot{\rho}^+ &= \dot{\rho} - \frac{( \dot{\bm{w}} \cdot \bm{q} )}{r}.
\end{aligned}
\end{equation}

The properties of instant manoeuvres can be used to model non-instant, gradual manoeuvres by discretisation. Any ownship
trajectory can be approximated as a series of instant manoeuvres interwoven with straight leg motion. The total
cumulative duration of the straight legs should equal the duration of the original complete manoeuvre. Using
sufficiently short legs, we can achieve the desired accuracy of approximation.

With all the preliminaries in place, we now begin to explore the state space distribution from a statistical
perspective.

\section{Closed-Form Expressions for the Post-Manoeuvre Distribution\label{sec:Distribution}}
\subsection{The Post-Manoeuvre Distribution\label{sec:NotGaussian}}
Gaussian approximations of the distribution of state space variables are ubiquitous in filtering, for example, when
using UKF~\cite{Wan:2000:Unscented} or ensemble methods~\cite{Houtekamer:1998:EnKF}. Assume the pre-manoeuvre
distribution of the state space variables is Gaussian. What about the post-manoeuvre distribution?

\begin{prop}[Post-manoeuvre distribution]
\label{prop:prelim_dist}
Let $f_x$ refer to the pre-manoeuvre probability distribution function for the target in coordinates $\bm{x}$ (either
MPC or LPC). Then, the post-manoeuvre distribution is given by
\begin{equation}
    f_{x^+}(\bm{x}^+) = f_x(\bm{h}^{-1}(\bm{x}^+)),
    \nonumber
\end{equation}
where
\begin{equation}
    \label{eq:MPCTransitionInverse}
    \bm{h}^{-1}(\bm{x}^+) =
    \begin{bmatrix}
        \beta^+\\
        \dot{\beta}^+ + f_1(\beta^+, \rho^+) \\
        \dot{\rho}^+ + f_2(\beta^+, \rho^+) \\
        \rho^+
    \end{bmatrix},
\end{equation}
and
\begin{subequations}
    \label{eq:f12}
    \begin{alignat}{2}
      f_1(\beta, \rho) &= e^{-\rho}(\Delta v_x \cos \beta - \Delta v_y \sin \beta), \\
      f_2(\beta, \rho) &= e^{-\rho}(\Delta v_x \sin \beta + \Delta v_y \cos \beta).
    \end{alignat}
\end{subequations}
If working in MPC, then the fourth coordinate of $\bm{x}^+$ is $s^+$ and factors $e^{-\rho}$ in $f_1$ and $f_2$ are
rewritten as $s$.
\end{prop}

\begin{proof}
Define the coordinate transformation by~\cref{eq:MPCTransitionM00} from the pre-manoeuvre distribution in $\bm{x}$ to
the post-manoeuvre distribution in $\bm{x}^+$, as the mapping $\bm{h}(\bm{x})$:
\begin{equation}
        \label{eq:MPCTransitionExpanded}
        \bm{h}(\bm{x}) =
        \begin{bmatrix}
            \beta\\
            \dot{\beta} - f_1(\beta, \rho)\\
            \dot{\rho} - f_2(\beta, \rho)\\
            \rho
        \end{bmatrix}.
\end{equation}
By observation, $\bm{h}^{-1}$ defined in the proposition is the inverse of $\bm{h}$.  The multivariate change of
variable theorem then expresses the post-manoeuvre distribution:
\begin{equation}
    \label{eq:MultivariateChangeOfVariable}
    f_{x^+}(\bm{x}^+) = f_x(\bm{h}^{-1}(\bm{x}^+)) \lvert \det\mathcal{J}(\bm{x}^+) \rvert,
\end{equation}
where $\mathcal{J}$ is the Jacobian matrix
\begin{equation*}
    \mathcal{J}(\bm{x}^+) =
    \begin{bmatrix}
        1 & 0 & 0 & 0 \\
        -f_2(\beta, \rho) & 1 & 0 & -f_1(\beta, \rho) \\
        f_1(\beta, \rho) & 0 & 1 &  -f_2(\beta, \rho) \\
        0 & 0 & 0 & 1 \\
    \end{bmatrix}.
\end{equation*}
It is easy to check that $\mathcal{J}$ has determinant $1$, so~\cref{eq:MultivariateChangeOfVariable} reduces to the
appropriate form.
\end{proof}
%%%%

As the transformation by~\cref{eq:MPCTransitionExpanded} does not alter the bearing $\beta^+$ or the range coordinate
($\rho^+$ or $s^+$), we anticipate that the marginal post-manoeuvre distribution of these two variables is unchanged
from the pre-manoeuvre distribution. Information about the target range comes from the entanglement of $\rho^+$ or $s^+$
with $\dot\beta^+$ and $\dot\rho^+$. We now lay out some special cases in which the post-manoeuvre distribution can be
found exactly.

\begin{prop} \label{prop:dist}
Assume that the pre-manoeuvre distribution (in either LPC or MPC) is well approximated by the Gaussian distribution
$\bm{x} \sim \mathcal{N}(\bm{\mu}, \Sigma)$, so \cref{prop:prelim_dist} provides
\begin{equation}
    \label{eq:MultivariateGaussian}
    f_{x^+}(\bm{x}^+) \propto \text{exp} \left\{ -\frac{1}{2} \left[ \bm{y}^{\text{T}} \Sigma^{-1}\bm{y} \right] \right\},
\end{equation}
where $\bm{y} = h^{-1}(\bm{x}^+) - \bm{\mu}$.
\begin{enumerate}[label=(\arabic{*}), ref=(\arabic{*})]
\item\label{mpc} If the coordinate system is MPC, $x_4 = s$, then
\begin{enumerate}[label=(\roman{*}), ref=(\roman{*})]
 \item\label{mpc.nn} the post-manoeuvre distribution is non-Gaussian, but
 \item\label{mpc.n} the post-manoeuvre distribution is Gaussian if conditioned on the bearing rate $\beta$. In this
 case, the post-manoeuvre distribution in MPC is
\[
    \begin{pmatrix}
        \dot\beta \\ \dot\rho \\ s
    \end{pmatrix}
    \sim
    \mathcal{N} \left(
        \begin{pmatrix}
            \mu_{\dot\beta} - c_1\mu_s \\
            \mu_{\dot\rho}-c_2\mu_s \\
            \mu_s
        \end{pmatrix},
        \mathbf{P}\Sigma_{-\beta} \mathbf{P}^{\text{T}}
    \right).
\]
In this distribution, $c_1= \Delta v_x \cos \beta - \Delta v_y \sin \beta$ and $c_2 = \Delta v_x \sin \beta + \Delta v_y
\cos \beta$ denote the constant factors in~\cref{eq:f12}, $\Sigma_{-\beta}$ means the covariance matrix $\Sigma$ with
the first row and column removed, and
\[
    \mathbf{P} = \begin{pmatrix}
        1 & 0 & -c_1\\
        0 & 1 & -c_2\\
        0 & 0 & 1
    \end{pmatrix}.
\]
It is important to recall that $s$ is strictly positive and any significant probability mass for $s<0$ should formally
be treated as follows: $s\rightarrow-s,\,\beta\rightarrow\beta\pm\pi$.
\end{enumerate}
\item If the coordinate system is LPC, $x_4 = \rho$, then
\begin{enumerate}[label=(\roman{*}), ref=(\roman{*})]
 \item the post-manoeuvre distribution is non-Gaussian, but
 \item the post-manoeuvre distribution is Gaussian if conditioned on $\beta$ and $\rho$, and in this case
\[
    \begin{pmatrix}
        \dot\beta \\
        \dot\rho
    \end{pmatrix}
    \sim
    \mathcal{N} \left(
        \begin{pmatrix}
            \mu_{\dot\beta} + f_1(\beta,\,\rho) \\
            \mu_{\dot\rho}+f_2(\beta,\,\rho)
        \end{pmatrix},
        \Sigma_{-\beta\rho}
    \right),
\]
where $\Sigma_{-\beta\rho}$ refers to the $2\times2$ covariance matrix formed by omitting the first and last row and
column of $\Sigma$.
\end{enumerate}
\end{enumerate}
\end{prop}
The proof of~\cref{prop:dist} is in~\cref{app:vis_man}-1. In order to understand the Proposition clearly, let us consider
an illustrative example. Consider the situation of a nearby, poorly estimated target. In Cartesian coordinates, the
pre-manoeuvre distribution estimates that the target is at $(100, 150)$~m from the ownship, but with a standard
deviation of $400$~m. Assume that $\beta$ is fairly accurately known from data. Then~\cref{fig:vis_man} shows how
unusual the MPC and LPC distributions appear in practice. The LPC distribution, which is the topic of this paper,
has many samples with great deviation from the sample mean; that is, the LPC distribution has large kurtosis.

\begin{figure}
  \centering
  \includegraphics[width=2.25in]{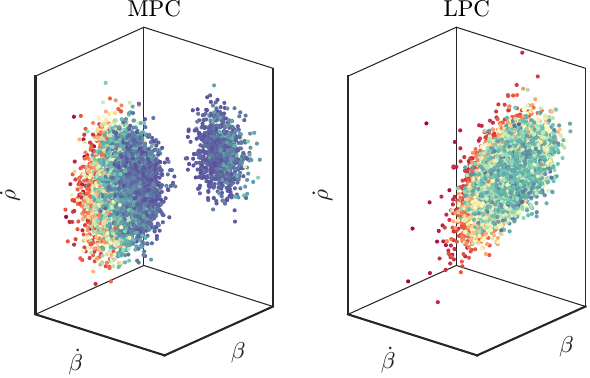}
  \caption{Visualisation by samples of the post-manoeuvre distribution in MPC and LPC. The colour of each sample
  indicates the estimated range of the target in $s$ or $\rho$. As described below~\cref{prop:dist}, the unusual
  situation in play is one in which the actual distance to the target is less than the \emph{uncertainty} in the
  target position. The MPC distribution consists of two roughly Gaussian blobs, which are separated by $\pi$ in
  $\beta$ and represent the probability mass on each side of the ownship. The LPC distribution has clear skew and
  kurtosis; that is, there are rare, large deviations from the sample mean. Notice in both plots that there is a clear
  correlation between $\dot\beta$ and the range, shown by the colour banding; as explained
  in~\cref{prop:dist}, this correlation has been created by the manoeuvre.}
  \label{fig:vis_man}
\end{figure}

Altering the initial position of the target to $(1000, 1500)$~m, and preserving all other parameters, we obtain roughly
identical Gaussian distributions for MPC and LPC in~\cref{fig:vis_man_far}. Full details for
both~\cref{fig:vis_man_far,fig:vis_man} are given in~\cref{app:vis_man}-2.

\begin{figure}
  \centering
  \includegraphics[width=2.25in]{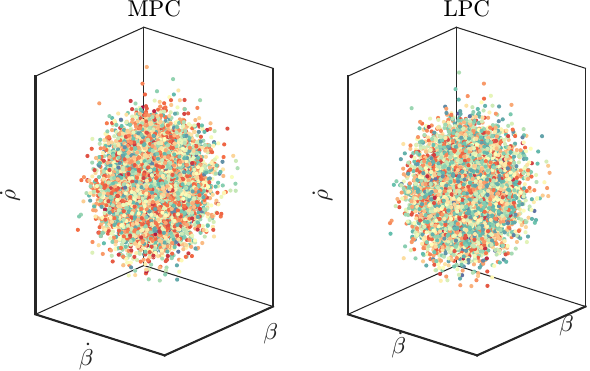}
  \caption{Visualisation by samples of the post-manoeuvre distribution in MPC and LPC. The colour of each sample
    indicates the estimated range. As described below~\cref{prop:dist}, this target is located far from the ownship. The
    MPC and LPC distribution are approximately identical.}
  \label{fig:vis_man_far}
\end{figure}

Our goal in investigating the post-manoeuvre distribution is to assess the viability of Kalman Filter approaches to
target tracking. Such approaches work by estimating the target mean and covariance. Our concern is that those estimates
can be inadequate or misleading if the underlying distribution has, for example, significant kurtosis.~\cref{prop:dist}
establishes some special cases in which we do understand the distribution, but~\cref{fig:vis_man,fig:vis_man_far}
demonstrate that, in the general case, the LPC post-manoeuvre distribution both can, and can not, resemble a Gaussian
distribution. In order to assuredly use a Kalman Filter in LPC in the general case, we need to know whether a Gaussian
approximation is defensible. For this reason, the main body of the paper derives the moments of the post-manoeuvre
distribution. The skew and kurtosis of the post-manoeuvre distribution allow monitoring whether
a Kalman Filter estimate can be trusted.

\subsection{Mean and Covariance of the Post-Manoeuvre Distribution\label{sec:MeanCovariance}}
Define the pre-manoeuvre mean and covariance of the target state vector as $\bm{\mu}_x$ and $\Sigma_x$; we seek the
corresponding post-manoeuvre mean and covariance, $\bm{\mu}^{+}_x$ and $\Sigma^{+}_x$.

From~\cref{eq:MPCTransitionExpanded}, observe that any post-manoeuvre state space variable $x_k^+$, where
$k \in \{0, 1, 2, 3\}$, is defined in terms of its pre-manoeuvre value, bearing, and range as
\begin{equation}
    \label{eq:GeneralPosterior}
    x_k^+(x_k) = x_k + a_k \frac{\sin \beta}{r} + b_k \frac{\cos \beta}{r},
\end{equation}
where $a_0 = a_3 = b_0 = b_3 = 0$,  $a_1 = \Delta v_y$ and $b_1 = -\Delta v_x$, $a_2 = -\Delta v_x$ and $b_2 = -\Delta
v_y$.

The following results present the exact mean and covariance for the post-manoeuvre distribution.
\begin{thm}[Mean of the post-manoeuvre distribution]
    \label{thm:PMMean}
    The expected value of a post-manoeuvre variable $x_k^+$ defined in~\cref{eq:GeneralPosterior} is:
    \begin{align}
        \label{eq:GeneralMean}
        \mathbb{E}\left[x_k^+\right] = \;
            & \mathbb{E}\left[ x_k \right]
                + e^{-\mu_{\rho} - \frac{1}{2}(\sigma_{\beta \beta}
                - \sigma_{\rho \rho})}(a_k \sin(\mu_{\beta} - \sigma_{\beta\rho}) \nonumber \\
            & + b_k \cos(\mu_{\beta} - \sigma_{\beta\rho})).
     \end{align}
\end{thm}

\begin{proof}
The integrals for the mean values take the form
\begin{equation}
\label{eq:PostTurn}
    \mathbb{E}\left[ x_k^+ \right] = \int x_k^+(\bm{x}) \mathcal{N}(\bm{x}; \bm{\mu}, \Sigma) d\bm{x}.
\end{equation}

For $k \in \{0, 3\}$ (bearing and log-range), $x_k^+ = x_k$, hence $\mathbb{E}\left[x^+_k\right] = \mu_k$.

Now consider the mean for bearing rate, $k = 1$. Substituting~\cref{eq:MPCTransitionM00} for $\dot{\beta}^+$
into~\cref{eq:PostTurn}, we get
\begin{align}
\label{eq:BetaDotPlusMean}
\mathbb{E}\left[ \dot{\beta}^+ \right]
    &= \int \dot{\beta}^+ \mathcal{N} \! \left( \bm{x}; \bm{\mu}, \Sigma \right) d\bm{x} \nonumber \\
    &= \int \left( \dot{\beta} - \frac{\dot{\bm{w}} \cdot \bm{n}}{r} \right) \mathcal{N} \! \left( \bm{x}; \bm{\mu}, \Sigma \right) d\bm{x} \nonumber \\
    &= \mathbb{E}\left[\dot{\beta}\right] - \dot{\bm{w}} \cdot \mathbb{E}\left[\frac{\bm{n}}{r}\right].
\end{align}
If we substitute the inverse range $r^{-1}$ with $e^{-\rho}$ and the sine and cosine components of $\bm{n}$ with real
and imaginary parts of $e^{i \beta}$, the last term becomes a linear combination of the real and imaginary parts
of~$\mathbb{E}\left[ e^{-\rho + i \beta} \right]$.

Note how this last expected value is similar to the moment-generating function (of the multivariate normal
distribution), just for a complex-valued argument. As the case of a complex argument appears unnamed in the literature,
we use the phrase `Generalised Gaussian MGF'.
\begin{prop}[The `Generalised Gaussian MGF' (GGMGF)]
    Let $\bm{h} \in \mathbb{C}^n$ be fixed and $\bm{x} \in \mathbb{R}^n$ be a Gaussian random variable. Then
    \label{prop:GeneralisedMGF}
    \begin{equation}
    \label{eq:GeneralMVG}
    \mathbb{E}\left[e^{\bm{h}^{\text{T}}\bm{x}}\right] = e^{\bm{\mu}^{\text{T}}\bm{h}+\frac{1}{2}\bm{h}^{\text{T}}\Sigma\bm{h}} \text{.}
    \end{equation}
\end{prop}
This generalises the Gaussian moment-generating function (MGF) to complex arguments.
The proof is described in \S 1.2 of \cite{Schmidt:2012:Stochastics} and \S 2 of \cite{Bryc:1995:Normal}.

\begin{cor}[GGMGF for specific $\bm{h}$]
    \label{cor:GeneralMGFExp}
    For any $\bm{h} = [in, 0, 0, -m]$, where $m, n\in\mathbb{N}\cup\{0\}$, the GGMGF has the form
    \begin{equation}
        \label{eq:ExpMN}
  \! \! \mathbb{E} \! \left[ r^{-m} e^{in\beta} \right] = e^{-m\mu_{\rho} - \frac{1}{2} \left( n^2\sigma_{\beta \beta} - m^2\sigma_{\rho \rho} \right) + i(n\mu_{\beta} - mn\sigma_{\beta\rho})}.
    \end{equation}
    The real and imaginary parts of \cref{eq:ExpMN} give $\mathbb{E} \left[ r^{-m} \cos\! {(n\beta)} \right]$ and
    $\mathbb{E} \left[ r^{-m} \sin\! {(n\beta)} \right]$, respectively.
\end{cor}

Now, continuing the proof, to obtain the expected value of $e^{-\rho + i \beta}$, we use \cref{cor:GeneralMGFExp} with
$\bm{h} = [i, 0, 0, -1]$ to find
\begin{equation}
    \label{eq:hInGGMGF}
 \! \mathbb{E} \! \left[ e^{-\rho + i \beta} \right] =
    \exp{ \! \big( \! -\mu_{\rho} - \frac{\sigma_{\beta \beta} - \sigma_{\rho \rho}}{2} + i(\mu_{\beta} -
    \sigma_{\beta\rho}) \big)}.
\end{equation}
The respective real and imaginary parts are $\mathbb{E}\left[r^{-1}\cos \beta\right]$ and $\mathbb{E}\left[r^{-1} \sin
\beta\right]$. Substitution of these expressions into~\cref{eq:BetaDotPlusMean} produces the mean of the post-manoeuvre
bearing rate:
\begin{multline}
    \label{eq:ExpectedBetaDot}
    \mathbb{E} \! \left[ \dot{\beta}^+ \right] =
      \mathbb{E} \! \left[ \dot{\beta} \right] -
          \left( \Delta v_x \cos(\mu_{\beta} - \sigma_{\beta\rho}) \right. - \\
          \left. \Delta v_y \sin(\mu_{\beta} - \sigma_{\beta\rho}) \right)
          e^{-\mu_{\rho} - \frac{1}{2} (\sigma_{\beta \beta} - \sigma_{\rho \rho})}.
\end{multline}
We can further simplify~\cref{eq:ExpectedBetaDot} by describing the ownship change in velocity through the turning angle
$\alpha$. Then, $\Delta v_x = \Delta v \cos \alpha$ and $\Delta v_y = \Delta v \sin \alpha$, so that
\begin{equation}
    \label{eq:ExpectedBetaDotAlpha}
\!\!  \mathbb{E} \! \left[ \dot{\beta}^+ \right] \! =
      \mathbb{E} \! \left[ \dot{\beta} \right] -
      \Delta v e^{-\mu_{\rho} - \frac{\sigma_{\beta \beta} - \sigma_{\rho \rho}}{2}} \cos{(\alpha + \mu_{\beta} - \sigma_{\beta\rho})}.
\end{equation}
Finally, consider scaled range rate, $k = 2$. Using~\cref{eq:MPCTransitionM00} for $\dot{\rho}^+$ in~\cref{eq:PostTurn},
and again employing \cref{prop:GeneralisedMGF}, we obtain
\begin{align}
    \mathbb{E} \! \left[ \dot{\rho}^+ \right]
      &= \int \dot{\rho}^+ \mathcal{N} \! \left( \bm{x}; \bm{\mu}, \Sigma \right) d\bm{x} \nonumber \\
      &= \int \left( \dot{\rho} - s(\dot{\bm{w}} \cdot \bm{q}) \right) \mathcal{N} \! \left( \bm{x}; \bm{\mu}, \Sigma \right) d\bm{x} \nonumber \\
      &= \mathbb{E} \left[ \dot{\rho} \right] - \dot{\bm{w}}\cdot \mathbb{E}\left[ e^{-\rho} \bm{q} \right] \nonumber \\
      &= \mathbb{E} \left[ \dot{\rho} \right] -
          \left( \Delta v_x \sin(\mu_{\beta} - \sigma_{\beta\rho}) \right. \nonumber \\
      & \quad \quad \left. + \Delta v_y \cos(\mu_{\beta} - \sigma_{\beta\rho}) \right)
      e^{-\mu_{\rho} - \frac{1}{2}(\sigma_{\beta \beta} - \sigma_{\rho \rho})} \label{eq:ExpectedRhoDot} \\
      &= \mathbb{E} \left[ \dot{\rho} \right]
          - \Delta v e^{-\mu_{\rho} - \frac{1}{2}(\sigma_{\beta \beta} - \sigma_{\rho \rho})} \sin{(\alpha + \mu_{\beta} - \sigma_{\beta\rho})}. \nonumber
\end{align}
\Cref{eq:ExpectedBetaDot,eq:ExpectedRhoDot} match \cref{eq:GeneralMean} with the coefficients $a_k$ and $b_k$
from~\cref{eq:GeneralPosterior}, thus completing the proof.
\end{proof}

Now derive the second raw moments for the post-manoeuvre distribution. By combining these second raw moments with the
first raw moments \labelcref{eq:GeneralMean}, we will also obtain the post-manoeuvre covariance matrix $\Sigma^+$.
\begin{thm}[Second raw moment of the post-manoeuvre distribution]
    \label{thm:SecondMoment}
    For any two post-manoeuvre state space variables $x_j^+$ and $x_k^+$, their second raw moment is:
    \begin{align}
        \label{eq:GeneralXiiXij}
        \mathbb{E} \! \left[ x_j^+ x_k^+ \right] = \;
          & \mathbb{E} \left[ x_j x_k \right] + a_k \mathbb{E} \! \left[x_j \frac{\sin{\beta}}{r} \right] + b_k \mathbb{E} \! \left[ x_j \frac{\cos{\beta}}{r} \right] \nonumber \\
          & + a_j \mathbb{E} \! \left[ x_k \frac{\sin{\beta}}{r} \right] + b_j \mathbb{E} \! \left[ x_j \frac{\cos{\beta}}{r} \right] \nonumber \\
          & + \frac{a_j a_k + b_j b_k}{2} \mathbb{E} \! \left[ \frac{1}{r^2} \right] \nonumber \\
          & + \frac{b_j b_k - a_j a_k}{2} \mathbb{E} \! \left[ \frac{\cos{\! (2\beta)}}{r^2} \right] \nonumber \\
          & + \frac{a_j b_k + a_k b_j}{2} \mathbb{E} \! \left[ \frac{\sin{\! (2\beta)}}{r^2} \right].
    \end{align}
    All the expected values in this expression can be obtained in closed form using the GGMGF~\labelcref{eq:GeneralMVG}.
\end{thm}

\begin{cor}[$k=j$ case of second raw moment]
    \label{cor:SecondMomentEqualKJ}
    If $k = j$ for the two state space variables in~\cref{thm:SecondMoment}, then
    \begin{align}
        \label{eq:GeneralXiiSquared}
        \mathbb{E} \! \left[ \left( x_j^+ \right)^2 \right] = \;
          & \mathbb{E} \! \left[x_j^2 \right] + 2 a_j \mathbb{E} \! \left[ x_j \frac{\sin{\beta}}{r} \right] + 2 b_j \mathbb{E} \! \left[ x_j \frac{\cos{\beta}}{r} \right] \nonumber \\
          &+ \frac{a_j^2 + b_j^2}{2} \mathbb{E} \! \left[ \frac{1}{r^2} \right] + \frac{b_j^2 - a_j^2}{2} \mathbb{E} \! \left[ \frac{\cos{\! (2\beta)}}{r^2} \right] \nonumber \\
          &+ a_j b_j \mathbb{E} \! \left[\frac{\sin{\! (2\beta)}}{r^2} \right].
    \end{align}
\end{cor}

Before presenting the proof of \cref{thm:SecondMoment}, we need one more auxiliary result. The theorem proof follows.

\begin{lem}[First derivative of GGMGF for specific $\bm{h}$]
    \label{lem:GeneralMGFXiExp}
    For any $\bm{h} = [in, 0, 0, -m]$, where $m, n\in\mathbb{N}\cup\{0\}$, the first derivative of the GGMGF with
    respect to $x_j$ has the form
    \begin{equation}
        \label{eq:PartialDerivative}
        \mathbb{E}\left[ x_j r^{-m} e^{in\beta} \right] = (\mu_j - m\sigma_{\rho j} + in\sigma_{\beta j}) \mathbb{E}\left[r^{-m} e^{in\beta} \right].
    \end{equation}
\end{lem}

\begin{proof}
    Consider the partial derivative of the GGMGF with respect to $h_j$. Swapping the order of expectation and
    differentiation, we have
\begin{align}
    \mathbb{E}\left[x_j e^{\bm{h}^{\text{T}} \bm{x}} \right]
    &= \frac{\partial}{\partial h_j} \left( \mathbb{E}\left[ e^{\bm{h}^{\text{T}} \bm{x}} \right] \right) \nonumber \\
    &= \frac{\partial}{\partial h_j} e^{\bm{\mu}^{\text{T}} \bm{h} + \frac{1}{2} \bm{h}^{\text{T}} \Sigma \bm{h}} \nonumber \\
    &= \left(\mu_j + \Sigma_j \bm{h} \right) e^{\bm{\mu}^{\text{T}} \bm{h} + \frac{1}{2} \bm{h}^{\text{T}} \Sigma \bm{h}}, \nonumber
\end{align}
where $\Sigma_j$ denotes the $j^\text{th}$ row of the matrix $\Sigma$. Without loss of generality,
substitute $\bm{h} = [in, 0, 0, -m]$ into this expression to obtain \cref{eq:PartialDerivative}.
\end{proof}

\begin{proof}[Proof of \cref{thm:SecondMoment}]
In the integral form, the second-order raw moments of the post-manoeuvre distribution are
\begin{equation}
    \label{eq:PostTurnM2}
    \mathbb{E} \! \left[ x_j^+ x_k^+ \right] =
      \int x_j^+ x_k^+ \mathcal{N} \! \left( \bm{x}; \bm{\mu}, \Sigma \right) d\bm{x}.
\end{equation}
Substituting \cref{eq:GeneralPosterior} for $x_j^+$, $x_k^+$, expanding the product, then interchanging the order of
integration and expectation, and grouping the terms with the same expected values we obtain the decomposition
\labelcref{eq:GeneralXiiXij}. All the expected values in~\cref{eq:GeneralXiiXij} are given in closed form by the real
and imaginary parts of ~\cref{eq:ExpMN,eq:PartialDerivative} for the respective $m$ and $n$ indices.
\end{proof}
\Cref{cor:SecondMomentEqualKJ} immediately follows from \cref{eq:GeneralXiiXij} for $j = k$. Therefore, we have derived
closed-form expressions for all the post-manoeuvre first and second (raw and central) moments.

\subsection{Higher-Order Moments of the Post-Manoeuvre Distribution\label{sec:HigherMoments}}
In~\cref{thm:PMMean,thm:SecondMoment}, we see that the first two moments of the post-manoeuvre distribution are linear
combinations of pre-manoeuvre moments. We generalise this for the arbitrary $N^{\text{th}}$ raw moment of the
post-manoeuvre distribution.

\begin{prop}[$N^\text{th}$ raw moment of the post-manoeuvre distribution]
    \label{prop:GeneralNthMoment}
    For post-manoeuvre state space variables $x_k^+$ each raised to powers $q_k$, $k = 0, \dots, 3$,
    such that $\sum_{j=0}^3 q_k = N$, the $N^\text{th}$ raw moment is
    \begin{align}
        \label{eq:NthMomentAsProduct}
 \!\!\! \mathbb{E} \! \left[ \prod_{k=0}^3 (x_k^+)^{q_k} \right] \! =
            \mathbb{E} \! \left[ \prod_{k=0}^3 \! \left(x_k + a_k \frac{\sin \beta}{r} + b_k \frac{\cos \beta}{r} \! \right)^{q_k} \right] \!\!,
    \end{align}
    which, when expanded, yields a linear combination of
    \begin{align}
        \label{eq:ExpectedNotation}
        E_{p_{0},p_{1},p_{2},p_{3},m,n}^{(s)} &= \mathbb{E}\left[\prod_{k=0}^3 x_k^{p_k} r^{-m}\sin (n \beta)\right]\text{, and} \nonumber \\
        E_{p_{0},p_{1},p_{2},p_{3},m,n}^{(c)} &= \mathbb{E}\left[\prod_{k=0}^3 x_k^{p_k} r^{-m}\cos (n \beta)\right],
    \end{align}
    where $\sum_{k=0}^3 p_k + m = N$, for $0 \le n \le m$.
\end{prop}

\begin{proof}
Consider the product
    \[
        \prod_{k=0}^3 \left(x_k + a_k \frac{\sin \beta}{r} + b_k \frac{\cos \beta}{r}\right)^{q_k}
    \]
from the right-hand side of~\cref{eq:NthMomentAsProduct}. Expansion via the binomial theorem yields a linear combination
of terms
\begin{equation}
\label{eq:GeneralPowerTerm}
\prod_{k=0}^3 x_k^{p_k} r^{-(c+d)} \sin^c{\beta} \cos^d{\beta},
\end{equation}
where $\sum_{k=0}^3 p_k + c + d = N$.

Using De Moivre's formula, Euler's formula, and the binomial theorem, we convert the powers $\sin^c{\beta}$ and
$\cos^d{\beta}$ to sums of $\sin(c^*\beta)$ and $\cos(d^*\beta)$, where $c^* = 0, \dots, c$ and $d^* = 0, \dots, d$.
Expanding again the product of these sums, and using the trigonometric addition identities, each
element~\labelcref{eq:GeneralPowerTerm} becomes a linear combination of terms
\begin{equation}
    \label{eq:GeneralMultTerm}
    \prod_{k=0}^3 x_k^{p_k} r^{-m} \sin \! {(n \beta)}, \quad \prod_{k=0}^3 x_k^{p_k} r^{-m} \cos \! {(n \beta)},
\end{equation}
where $m = c + d$, $\sum_{k=0}^3 p_k + m = N$, and $0 \le n \le m$.

Therefore, using the linearity of expectation, the right-hand side of~\cref{eq:NthMomentAsProduct} is a linear
combination of $E_{p_{0},p_{1},p_{2},p_{3},m,n}^{(s)}$ and $E_{p_{0},p_{1},p_{2},p_{3},m,n}^{(c)}$ defined
in~\cref{eq:ExpectedNotation}, which are the expected values of expressions~\labelcref{eq:GeneralMultTerm}.
\end{proof}

The expected values in~\cref{eq:ExpectedNotation} can be found in closed form by recursive application of the following
Lemma.

\begin{lem}[$N^\text{th}$ partial derivative of the GMGF]
    \label{lem:NthPartialDerivative}
    The $N^{\text{th}}$ partial derivative of the GMGF, conditioned by $\sum_{k=0}^3 p_k + m = N$, $p_j \ge 0$, and $m
    \ge 0$, is
    \begin{subequations}
    \label{eq:ExpectedStateExp}
    \begin{align}
        \mathbb{E}\left[ e^{\bm{h}^{\text{T}} \bm{x}} \prod_{k=0}^3 x_k^{p_k} \right]
            = &\frac{\partial^{N} \Phi_{\bm{x}}(\bm{h})}{\partial h_j^{p_j} \partial h_k^{p_k} \dots } \nonumber \\
            = &\sum_{j \ne k} p_j \sigma_{jk} \frac{\partial^{N-2} \Phi_{\bm{x}}(\bm{h})}{\partial h_j^{p_j-1} \partial h_k^{p_k-1} \dots} \label{eq:ExpectedStateExp:1} \\
              &+(p_k - 1) \sigma_{kk} \frac{\partial^{N-2} \Phi_{\bm{x}}(\bm{h})}{\partial h_j^{p_j} \partial h_k^{p_k-2} \dots} \label{eq:ExpectedStateExp:2} \\
              &+ \left(\mu_k + \Sigma_k \bm{h} \right) \frac{\partial^{N-1} \Phi_{\bm{x}}(\bm{h})}{\partial h_j^{p_j} \partial h_k^{p_k-1} \dots} \label{eq:ExpectedStateExp:3},
    \end{align}
    \end{subequations}
    where $k$ is selected such that $p_k > 1$, which is always possible for $N \ge 1$. The index $k$ is not specific,
    any index can be selected due to the symmetry of derivatives.

    For the sake of brevity, the notation in \cref{eq:ExpectedStateExp} technically allows negative derivatives of order
    $-1$ if $p_j = 0$ for some elements in \cref{eq:ExpectedStateExp:1}, or if $p_k = 1$ in
    \cref{eq:ExpectedStateExp:2}. However, such terms are eliminated by the corresponding $p_j = 0$ and $p_k - 1$
    multipliers in front.
\end{lem}

\begin{proof}
For $N = 1$, we re-write~\cref{eq:PartialDerivative} as
\begin{equation}
    \label{eq:FirstPartialDerivative}
    \frac{\partial}{\partial h_k} \Phi_{\bm{x}}(\bm{h})
        = \left(\mu_k + \Sigma_k \bm{h} \right) \Phi_{\bm{x}}(\bm{h}).
\end{equation}
For $N = 2$, we take the second partial derivative to obtain
\begin{align}
    \label{eq:SecondPartialDerivative}
    \frac{\partial^{2}}{\partial h_k \partial h_j} \Phi_{\bm{x}}(\bm{h})
        &= \frac{\partial}{\partial h_k} \left(\mu_j + \Sigma_j \bm{h} \right) \Phi_{\bm{x}}(\bm{h}) \nonumber \\
        &= \bigl(\sigma_{jk} + \left(\mu_j + \Sigma_j \bm{h} \right) \left(\mu_k + \Sigma_k \bm{h} \right) \bigr) \Phi_{\bm{x}}(\bm{h}) \nonumber \\
        &= \resizebox{0.55\hsize}{!}{%
        $\sigma_{jk} \Phi_{\bm{x}}(\bm{h}) + \left(\mu_k + \Sigma_k \bm{h} \right) \frac{\partial}{\partial h_j} \Phi_{\bm{x}}(\bm{h})$}.
\end{align}
These expressions start the recursion of \cref{eq:ExpectedStateExp}: for $j \ne k$, the first term of
\cref{eq:SecondPartialDerivative} yields \cref{eq:ExpectedStateExp:1}, and for $j = k$, it yields
\cref{eq:ExpectedStateExp:2}. \cref{eq:FirstPartialDerivative} and the last term of \cref{eq:SecondPartialDerivative}
produce \cref{eq:ExpectedStateExp:3} for the respective $N = 1, 2$.

The proof for a general order $N$ is obtained through direct differentiation of \cref{eq:ExpectedStateExp} by $h_m$ for
some arbitrary $m$: the partial derivatives in \crefrange{eq:ExpectedStateExp:1}{eq:ExpectedStateExp:3} increase the
order of their $h_m$-th derivatives by one. Then the derivative of the $\mu_k + \Sigma_k \bm{h}$ term in
\cref{eq:ExpectedStateExp:3} yields one more summand with a $\sigma_{km}$ multiplier. If $k \ne m$, it is combined with
the corresponding $j = m$ term of \cref{eq:ExpectedStateExp:1} to increase its $p_m$ multiplier in front to $p'_m = p_m
+ 1$. Alternatively, for $k = m$ the new summand is combined with \cref{eq:ExpectedStateExp:2} to increase the $p_k - 1
= p_m - 1$ multiplier to $p_m = p'_m - 1$. In both cases, we obtain \cref{eq:ExpectedStateExp} with substitutions $N
\Rightarrow N' = N + 1$ and $p_m \Rightarrow p'_m = p_m + 1$, which is the next order of recursion.
\end{proof}

Substituting $\bm{h} = [in, 0, 0, -m]$ into \cref{eq:ExpectedStateExp} recursively yields the expected values
from~\cref{eq:ExpectedNotation}.

By \cref{prop:GeneralNthMoment}, any $N^\text{th}$ order raw moment is
\begin{align}
    \label{eq:NthMomentAsSum}
    \mathbb{E}\left[\prod_{k=0}^3 (x_k^+)^{q_k}\right] &=
    \sum_{\{q,m\}} \sum_{n=0}^m \left(S_{p_{0}, \dots, p_{3},m,n}^{(p_0, \dots, p_3)} E_{p_{0}, \dots, p_{3},m,n}^{(s)} \right. \nonumber \\
    & \left. + C_{p_{0}, \dots, p_{3},m,n}^{(p_0, \dots, p_3)} E_{p_{0}, \dots, p_{3},m,n}^{(c)} \right).
\end{align}
The first sum considers all possible combinations of $p_k$, $k=0,\dots,3$ and $m$ subject to the conditions
$\sum_{k=0}^3 p_k + m = N$, $p_k \ge 0$, $m \ge 0$. The coefficients $S_{p_{0}, \dots, p_{3},m,n}^{(p_0, \dots, p_3)}$
and $C_{p_{0}, \dots, p_{3},m,n}^{(p_0, \dots, p_3)}$ are products of $a_k$ and $b_k$ and can be found
from~\cref{eq:NthMomentAsProduct,eq:GeneralPosterior}.

We demonstrated that any higher-order moment of the post-manoeuvre distribution can be found in closed form by recursive
application of~\cref{eq:ExpectedStateExp,eq:NthMomentAsSum}. However, as the order of the moment increases, so does the
complexity of its closed-form expression. Explicit calculation of the coefficients $S_{p_{0}, \dots, p_{3},m,n}^{(p_0^+,
\dots, p_3^+)}$ and $C_{p_{0}, \dots, p_{3},m,n}^{(p_0^+, \dots, p_3^+)}$ quickly becomes prohibitive. Combined with the
non-Gaussianity proven in~\cref{sec:NotGaussian}, such that we have no existing statistical results to draw on, we turn
to computer algebra for deriving specific cases of~\cref{eq:NthMomentAsSum}. We have developed a package in SageMath
\cite{TheSageDevelopers:2023:SageMath} to generate these moments and export them as LaTeX or Python code. A
demonstration of this package is shown in \cref{app:ExampleMoment}.

The first two moments are explicitly given by~\cref{eq:GeneralMean,eq:GeneralXiiXij}. In the next section, we apply
these two moments in a modified UKF estimator, while using the higher-order moments to monitor its performance.

\section{Applications of the Closed-Form Expressions in Tracking Simulation\label{sec:Applications}}
\subsection{UKF using Closed-Form Expressions for Mean and Covariance Prediction\label{sec:ModifiedUKF}}
The Kalman Filter provides a minimum square error estimate for the mean and covariance of the target
location~\cite{Labbe:2014:Kalman}. For our non-linear dynamical system, the UKF is appropriate. It uses the
UT to propagate deterministic sigma points that accurately capture the target state mean and
covariance~\cite{Wan:2000:Unscented}.

In order to test our closed-form expressions in tracking simulation, we will implement a pure UKF and a modified UKF
that incorporates our closed-form expressions. We assume the ownship alternates between straight legs and instant turns.
The modified UKF employs the first two moments of the post-manoeuvre distribution to calculate the target state and
covariance after each ownship turn.

We define the ownship manoeuvre at time step $i$ as a tuple of three values:
\begin{itemize}
    \item $\Delta \bm{r}_i$, the accelerated displacement during the step illustrated in \cref{fig:ThreeStepManoeuvre}.
    \item $\Delta \bm{v}_i$, the speed change during the step; the speed changes instantly at the point of the turn, at
        unknown time within the step.
    \item $\Delta t_i$, the time duration of the step.
\end{itemize}
As the turn may happen at a time instance within the step, we seek to split the manoeuvre into elementary movements, or
sub-manoeuvres: a straight leg, an instant turn, and another straight leg after the turn. The duration of the two
straight legs adds up to the total duration of the time step $\Delta t_i$. Computation of these sub-manoeuvres is
described in~\cref{thm:SplitManoeuvre}.

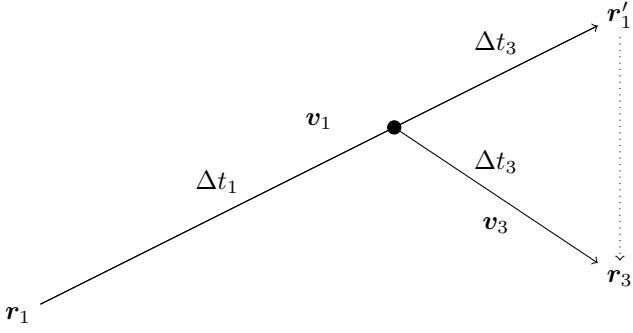
\begin{figure}[!t]
    \centering
\begin{tikzpicture}
    \node (r1) at (0,0) {$\bm{r}_1$};
    \node (r1hat) at (8,4) {$\bm{r}_1'$};
    \node (r3) at (8, 0.5) {$\bm{r}_3$};
    \node[circle, draw=black, fill=black, inner sep=0pt, minimum size=5pt] (b) at (5,2.5) {};
    \draw (r1) -- (r1hat) node [midway, above=10pt] {$\bm{v}_1$};
    \draw [->] (r1) -- (5,2.5) node [midway, above=5pt] {$\Delta t_1$};
    \draw [->] (5,2.5) -- (r1hat) node [midway, above=5pt] {$\Delta t_3$};
    \draw [->] (5,2.5) -- (r3) node [midway, above=5pt] {$\Delta t_3$} node [midway, below=5pt] {$\bm{v}_3$};
    \draw[dotted] [->] (r1hat) -- (r3);
  \end{tikzpicture}
  \caption{A diagram of the predicted ownship position, $\bm{r}_1'$, versus the true position, $\bm{r}_3$, where the
    black node indicates the point of the instant turn. Note the jump from index $1$ to $3$ as the second sub-manoeuvre,
    the instant turn, has zero duration, zero accelerated displacement (hence, zero length on the plot), but non-zero
    change in velocity: $\Delta \bm{v} = \bm{v}_3 - \bm{v}_1$. The accelerated displacement is $\Delta \bm{r} = \bm{r}_3
    - \bm{r}_1'$.}
  \label{fig:ThreeStepManoeuvre}
\end{figure}

\begin{thm}[Manoeuvre splitting]
    \label{thm:SplitManoeuvre}
    The manoeuvre containing a turn $M(\Delta \bm{r}, \Delta \bm{v}, \Delta t)$ is equivalently described as a list of
    three sub-manoeuvres:
\begin{enumerate}
    \item $M_1(0, 0, \Delta t_1)$,
    \item $M_2(0, \Delta \bm{v}, 0)$, and
    \item $M_3(0, 0, \Delta t_3)$
\end{enumerate}
where
\begin{equation}
    \label{eq:DeltaT3FromMinimisation}
    \Delta t_3 = \frac{ \Delta \bm{r} \cdot \Delta \bm{v} }{|\Delta \bm{v}|^2}
\end{equation}
and $\Delta t_1 = \Delta t - \Delta t_3$.
\end{thm}

\begin{proof}
    We write the final position in terms of the initial position and two velocities.
    \[
        \bm{r}_3 = \bm{r}_1 + \Delta t_1 \bm{v}_1 + \Delta t_3 \bm{v}_3.
    \]
    Substituting $\Delta t_3 = \Delta t - \Delta t_1$ and regrouping, we get
    \[
        \bm{r}_3 - (\bm{r}_1 + \Delta t \bm{v}_1) = \Delta t_3 (\bm{v}_3 - \bm{v}_1).
    \]
    As illustrated in~\cref{fig:ThreeStepManoeuvre}, the left-hand side is the accelerated displacement $\Delta \bm{r}$;
    the term in brackets in the right-hand side is the speed change $\Delta \bm{v}$. Therefore,
    \[
        \Delta \bm{r} = \Delta t_3 \Delta \bm{v}.
    \]
    By assumption of a non-zero turn, $|\Delta \bm{v}| \ne 0$. Thus, multiplying both sides by $\Delta \bm{v}$ we obtain
    \cref{eq:DeltaT3FromMinimisation}.
\end{proof}

In cases where the turn is at the start or at the end of the total manoeuvre, without loss of generality, we split the
total manoeuvre into two: a turn and a straight leg.

Having established the approach for isolating instant turns, we present the modified filter, CFE-UKF,
as~\cref{alg:ModifiedUKF}.
\begin{algorithm}[!ht]
    \caption{\label{alg:ModifiedUKF} The modified UKF algorithm, \texttt{CFE-UKF}, that uses the closed-form expressions
    \labelcref{eq:GeneralMean,eq:GeneralXiiXij} to update the mean and covariance over instant turns.}
    \begin{algorithmic}[1]
        \Function{CFE-UKF}{}
            \State Initialise at time $t_0$, iteration $0$
            \State $N_T$ is the total count of time steps
            \For{$i \in [1, N_T]$}
                \State Advance platforms forward by time step $\Delta t_i$
                \State Compute the manoeuvre $M_i(\Delta \bm{r}_i, \Delta \bm{v}_i, \Delta t_i)$
                \If{$|\Delta \bm{v}_i| < \epsilon$} \Comment{no turn}
                    \State Time update over $\Delta t_i$ using UT
                \Else
                    \State Split $M_i$ into sub-manoeuvres $[M_j]$
                    \ForAll{$M_j \in [M_j]}$
                        \If{$\Delta \bm{v}_j \neq \bm{0}$} \Comment{ownship turned}
                            \State State update by~\cref{eq:GeneralMean,eq:GeneralXiiXij}
                        \Else \Comment{straight leg}
                            \State Time update over $\Delta t_j$ using UT
                        \EndIf
                    \EndFor
                \EndIf
                \State Measurement update of the \texttt{UKF} as usual
            \EndFor
        \EndFunction
    \end{algorithmic}
\end{algorithm}

\begin{figure}[!t]
    \centering
    \includegraphics[width=3.25in]{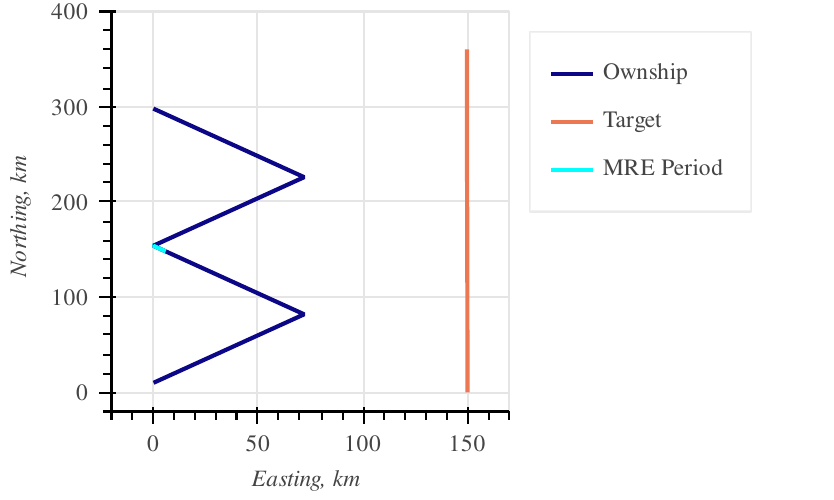}
    \caption{Motion of the ownship and target in the simulated scenario. In Cartesian coordinates, the ownship velocity
    along straight legs is $[\pm10, 10]$~m/s and the initial target velocity is $[0, 12.5]$~m/s. Each straight leg is
    2~hours long. The short cyan segment shows the stretch of time used for mean range error computations in
    \cref{sec:MomentMonitoring}.}
    \label{fig:OTScenario}
\end{figure}

To compare the CFE-UKF with a pure UKF, we use a scenario presented in \cref{fig:OTScenario}. The ownship moves
deterministically along straight lines, and makes exact turns, while the target movement is subject to the process noise
provided by the continuous-time white noise model with nearly constant velocity~\cite{Bar-Shalom:2001:Estimation}. The
time increment $\Delta t$ is 10~seconds.

For the purposes of verification, we first run the CFE-UKF and UKF with artificially precise parameters.
The process noise intensity is set to $10^{-8}$ and the sensor standard
deviation of bearing measurements to 0.001\degree. We set the initial track range to $150$~km (the true range from
ownship to target) and initialise the range variance to $10$\% of the true range. In this
idealised case, we find both trackers run very similarly: the averaged absolute difference between their means
is only $[1.3 \cdot 10^{-9}, 1.8 \cdot 10^{-11}, 3.9 \cdot 10^{-10}, 1.8 \cdot 10^{-5}]$.

The CFE-UKF is slightly less efficient than a pure UKF; for each turning point, we perform up to two UTs, depending on
whether there is a preceding \textit{and} succeeding straight leg, and then additionally compute the closed-form
expressions to update the mean and covariance over the turn. However, due to our small state space dimension $n_x = 4$,
and the small proportion of ownship turns relative to the total scenario, this increase in complexity has minimal
effect on total run-time.

\subsection{Monitoring Gaussianity Using the Third and Fourth Central Moments of the Post-Manoeuvre
Distribution\label{sec:MomentMonitoring}}
In~\cref{sec:NotGaussian} we have shown that the post-manoeuvre distribution is not Gaussian. Now we can quantify this
statement. When the ownship turns, we know the higher-order moments of the post-manoeuvre state distribution, and can
use them to monitor how close is this distribution to Gaussian. The turning points are critical for passive BO tracking
because only there the range becomes observable, and therefore, it is vital for the estimator assumptions, \textit{e.g.}
Gaussianity, to be satisfied.

We consider two metrics for use as a monitoring tool. The first it is the Frobenius norm of the third central moment
tensor
\begin{equation}
    \label{eq:M3}
    M_3 = \big\lVert \, \mathbb{E} \left[ (\bm{x}^+ - \bm{\mu}^+)^3 \right] \! \big\rVert.
\end{equation}
For a Gaussian distribution $M_3 \equiv 0$, hence a non-zero value is a measure of non-Gaussianity.

The second metric is the norm of the deviation of the post-manoeuvre fourth central moment, a $4^\text{th}$-order
tensor, from the expected fourth
central moment of a Gaussian distribution:
\begin{equation}
    \label{eq:M4}
    M_4 = \big\lVert \, \mathbb{E} \! \left[ (\bm{x}^+ - \bm{\mu}^+)^4 \right] - \mathbb{E} \! \left[ (\bm{y} - \bm{\mu}^+)^4 \right] \! \big\rVert,
\end{equation}
where $y \sim \mathcal{N}(\bm{\mu}^+, \Sigma^+)$. Again, $M_4 \equiv 0$ for a Gaussian distribution; a non-zero value
quantifies non-Gaussianity.
% For the future work: it would be nice to normalise the metrics somehow. For $M_4$, one option is to divide by the
% norm of the Gaussian fourth-order moment, which is guaranteed to be non-zero. There are also multi-dimensional
% \textit{skewness} and \textit{kurtosis} by Mardia. I'm not sure if those can be easily found from the available
% higher-order moments.

\Cref{fig:MomentTimeSeries} shows the time series of the proposed metrics for the scenario of~\cref{fig:OTScenario}. To
make the simulations somewhat less idealised, from now on the process noise intensity is increased to $10^{-6}$ and the bearing
measurement error to~0.1\degree. Clearly, the distribution is not Gaussian when the ownship turns. On the
positive side, we observe that with every subsequent turn the metrics decrease, implying that the distribution becomes
more Gaussian with time as the estimators converge. Consequently, it may imply that the observed non-Gaussianity is
related to the estimator initialisation, and that is what we examine next.

\begin{figure}[!t]
    \centering
    \includegraphics[width=3.5in]{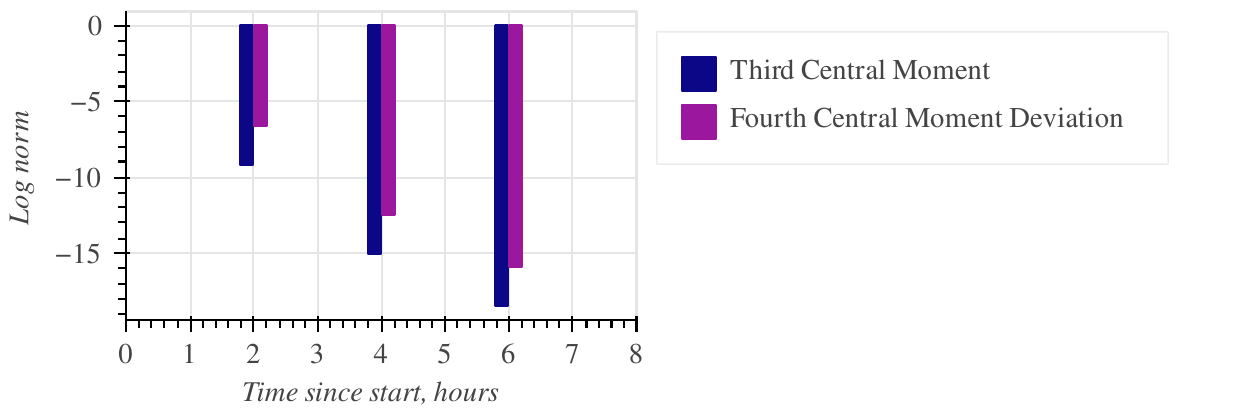}
\caption{Time series of the two non-Gaussianity metrics. The three peaks correspond to the three ownship turns (bars
have been offset for visibility).}
\label{fig:MomentTimeSeries}
\end{figure}

\subsection{Effect of Range Error Initialisation on the Post-Manoeuvre Distribution
\label{sec:RangeInitialisation}}
To start the estimator, we need to provide the initial state $\bm{x}_0$ and covariance $\Sigma_0$. In the examples
presented in this paper, we set $\bm{x}_0$ using the first meaured bearing for $\beta_0$, zeros for the derivatives
$\dot{\beta}_0$ and $\dot{\rho}_0$, and the true value for the initial range $r_0$. The covariance $\Sigma_0$ is set to
a diagonal matrix of
\begin{itemize}
    \item the sensor measurement variance for bearing,
    \item ${(v_{\text{max}} / r_0)}^2$ for the bearing and scaled range rates, where $v_{\text{max}}$ is some ``large''
        speed value, and
    \item $\sigma_{rr} = \log{(1 + {(\delta r)}^2)}$ for log-range, were we introduce $\delta r$ as the initial relative
        range error.
\end{itemize}
We express the log-range error in terms of $\delta r$, because in the limit of small errors, the range variance of
${(r_0 \delta r)}^2$ indeed corresponds to the log-range variance of $\sigma_{rr}$. For larger errors the relationship breaks
down, but $\delta r$ still provides a convenient description of the error in human terms, as opposed to the less
intuitive variance of the range logarithm.

To examine the impact of initialisation on the tracking accuracy and on the properties of the post-manoeuvre
distribution, we consider the estimator output after the first manoeuvre of our scenario, where the non-Gaussianity is
most pronounced. As a measure of accuracy, we use the mean estimated range error (MRE) calculated over ten minutes of
the scenario at the end of the second leg (highlighted in \cref{fig:OTScenario} with cyan).

\begin{figure}[!t]
    \centering
    \begin{subfigure}{3.5in}
        \hspace{-0.2cm} \centering
        \includegraphics[width=3.5in]{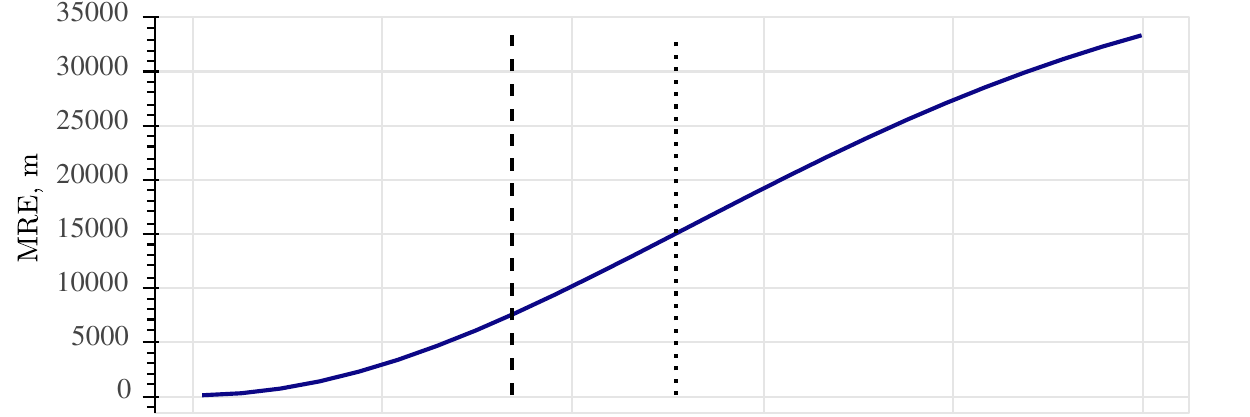}
        \phantomsubcaption{\label{fig:RangeErrorMRE}}
    \end{subfigure}\vspace{-0.4cm}%

    \begin{subfigure}{3.5in}
        \centering
        \includegraphics[width=3.5in]{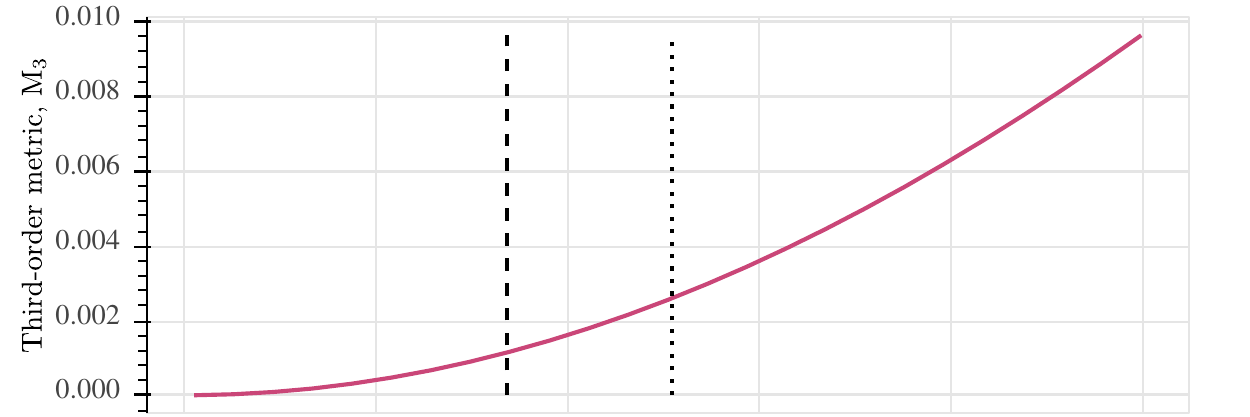}
        \phantomsubcaption{\label{fig:TMEffect}}
    \end{subfigure}\vspace{-0.4cm}%

    \begin{subfigure}{3.5in}
        \centering
        \includegraphics[width=3.5in,height=1.4in]{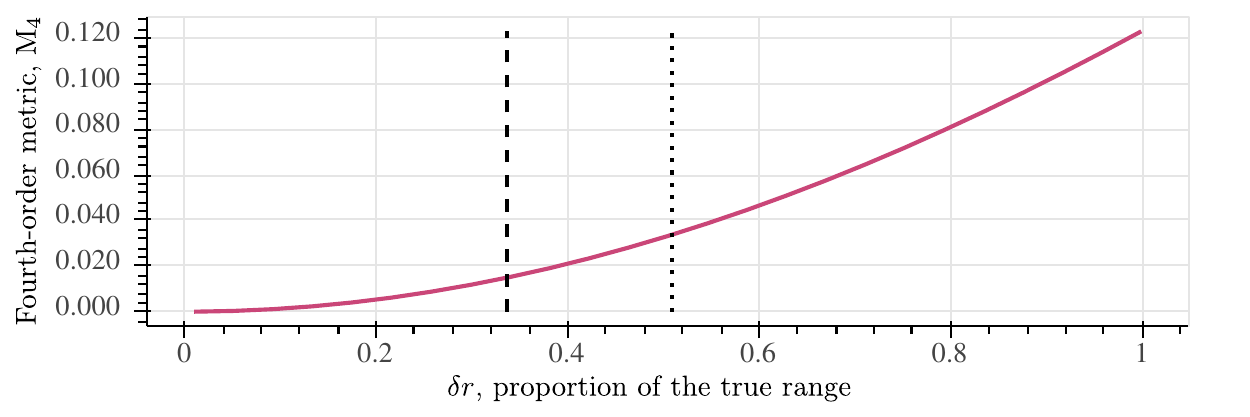}
        \phantomsubcaption{\label{fig:FMEffect}}
    \end{subfigure}

    \caption{Effect of the initial range variance on MRE (a, top), third central moment (b, middle), and fourth central
    moment difference (c, bottom). The dashed and dotted lines respectively show the thresholds to achieve the MRE of
    $5\%$ and $10\%$ of the initial range.}
    \label{fig:RangeErrorRelationships}
\end{figure}

The dependence of the MRE on tracker initialisation is presented in \cref{fig:RangeErrorMRE}. The corresponding two
metrics of target state Gaussianity after the first ownship turn are shown in \cref{fig:TMEffect,fig:FMEffect}. As
expected, estimation error grows with initialisation ambiguity.  Correspondingly, the post-manoeuvre state distribution
becomes less Gaussian, thus violating the UKF assumptions and reducing estimator accuracy.

A corollary of the presented result is that monitoring non-Gaussianity enables a measure of control of the estimator
performance. This is the motivation for the CFE-UKF. While it may not perform that differently from the pure UKF, the
new estimator provides the user with a capability to monitor whether the underlying assumptions of the Gaussian
estimator are being met.

Whenever the non-Gaussianity metrics exceed certain thresholds\footnote{ The specific thresholds for $M_3$ and $M_4$
metrics can be established by Monte-Carlo modelling, but are beyond the scope of this paper.}, the user would know that
the target range estimate after the manoeuvre cannot be trusted, even if the estimated range variance is low. In such a
case, the ownship may, for example, attempt the next manoeuvre to improve the solution. \Cref{fig:MomentTimeSeries}
shows how Gaussianity of the post-manoeuvre distribution improves on each turn, thus validating the estimator
convergence. Alternatively, the same Gaussianity metrics can guide the choice of the initial range variance $\delta r$.

The focus of this section has been the use of our results in a Kalman Filter framework. One may consider more powerful
filtering algorithms that do not assume Gaussianity of the target state; these are discussed in the next section.

\section{Conclusions and Future Work\label{sec:Conclude}}

The paper extends the theoretical understanding of object tracking in LPC. In order to make the theory tractable, we
assume an instantaneous manoeuvre and a target moving in a straight line. Under these simplifications, we obtain some
strong results. The main contribution of the paper is~\cref{prop:GeneralNthMoment}: we derive all moments of the
post-manoeuvre distribution in closed form. Additionally,~\Cref{prop:dist}~derives, for both LPC and MPC, special cases
in which the post-manoeuvre distribution is Gaussian, with known parameters.  Our results, while formally obtained for
instantaneous manoeuvres, may be applied in the general, non-instantaneous case through discretisation.

As the contributions of the paper are primarily theoretical, let us now discuss possible applications of the theory. The
most common class of tracking algorithms, consisting of extensions of the Kalman Filter for non-linear system dynamics,
work best when the state distribution of the target is approximately unimodal and without large outliers. Our
contribution allows a user to track the third and fourth moments of the LPC distribution, and thereby assess when these
conditions are met.

We demonstrate the utility of the obtained results in filtering algorithms. For state and covariance prediction, we
define a new algorithm, the CFE-UKF, that uses the closed-form expressions for mean and covariance in place of the
UT-based time update over an instant ownship turn. By comparing with the well-understood UKF in a simple tracking
scenario, we confirm our results for the first two moments of the post-manoeuvre distribution are correct.  However, in
addition to the mean and covariance, our estimator also computes the third and fourth post-manoeuvre moments to allow
monitoring of the quality of the Gaussian approximation, an ability that the pure UKF lacks. This insight gives the
CFE-UKF an edge over the UT-based estimator.

With this understanding, we conclude that there is a need to evaluate algorithms that do not blindly assume Gaussianity
of the underlying target state distribution. Similar to the approaches presented in \cite{Zanetti:2018:Novel}, we could
use our higher-order moment metrics as splitting criteria for a Gaussian sum filter; a threshold established by
Monte-Carlo modelling can determine when the estimator would benefit from a split of the prior when the ownship turns.

Alternatively, the conditional Gaussianity properties proved in \cref{prop:dist} empower the use of a Rao-Blackwell-type
particle filter; a particle filter would track the bearing rate and scaled range rate, while bearing and log-range are
conditionally Gaussian and can be tracked with a Kalman Filter.

\section*{Acknowledgements}
We thank Acacia Systems for their support and funding.

\begin{appendices}
\section{Supplemental details for~\cref{sec:NotGaussian}}
\label{app:vis_man}
\subsection*{1. Proof of~\cref{prop:dist}}
\begin{proof}
Claims (1)(i) and (2)(i), that the post-manoeuvre distributions of MPC and LPC are non-Gaussian, are clearly visible
from the terms in $\cos \beta$ and $\sin \beta$, which are entangled with $\dot\beta$ and $\dot\rho$ in the exponent of
the post-manoeuvre distribution. Let $P_{i,j}$ refer to the $i,j^\text{th}$ entry of the precision matrix $\Sigma^{-1}$; then
the exponent of~\cref{eq:MultivariateGaussian} contains a term
\[(\dot\rho - \mu_\rho + s\Delta v_x \sin \beta + s\Delta v_y \cos \beta)^2 P_{2,2}\;. \]
Expanding this quadratic reveals terms like, for example, $\dot\rho s P_{2,2} \sin\beta $, that render the distribution
non-Gaussian. One might wonder whether these terms can cancel with other terms in the distribution --- but cancellation
is impossible without requiring special conditions on the precision matrix entries.

Claim (1)(ii), that the MPC distribution is known exactly if we condition on the bearing $\beta$, is proven as follows.
As in the Proposition, let $c_1$ and $c_2$ refer to the constant factors in $f_1$ and $f_2$.
Accordingly,~\cref{eq:MPCTransitionInverse} in MPC becomes
\[
\bm{h}^{-1}(\bm{x}^+) = \begin{bmatrix}
                        \beta^+\\
                        \dot{\beta}^+ + c_1s^+\\
                        \dot{\rho}^+ + c_2s^+\\
                        s^+
    \end{bmatrix}\;,
\]
and on substitution into~\cref{eq:MultivariateGaussian} and conditioning on $\beta$, the exponent of the post-manoeuvre
distribution is
\[-\frac{1}{2}\left[
\left(
\begin{bmatrix}
                        \dot{\beta}^+ + c_1s^+\\
                        \dot{\rho}^+ + c_2s^+\\
                        s^+
    \end{bmatrix} - \bm{\mu}
\right)^{^{\text{T}}}
\Sigma^{-1}_{-\beta}
\left(
\begin{bmatrix}
                        \dot{\beta}^+ + c_1s^+\\
                        \dot{\rho}^+ + c_2s^+\\
                        s^+
    \end{bmatrix} - \bm{\mu}
\right)
\right]. \]
    We now write the vectors that pre- and post-multiply $\Sigma$ as linear transformations of Gaussian innovations. Let
    $d_3 = s^+ - \mu_s$, then
\[ \dot{\beta}^+ + c_1s^+ - \mu_{\dot\beta} = \dot\beta - (\mu_{\dot\beta} - c_1\mu_s) + c_1d_3,  \]
\[ \dot{\rho}^+ + c_2s^+ - \mu_{\dot\rho} = \dot\rho - (\mu_{\dot\rho} - c_2\mu_s) + c_2d_3.  \]
Accordingly, let $d_1 = \dot\beta - (\mu_{\dot\beta} - c_1\mu_s)$ and $d_2 = \dot\rho - (\mu_{\dot\rho} - c_2\mu_s)$,
then the exponent simplifies to
\[-\frac{1}{2}\left[
\begin{pmatrix} d_1+c_1d_3\\d_2+c_2d_3\\d_3 \end{pmatrix}
^{\text{T}}
\Sigma^{-1}
\begin{pmatrix} d_1+c_1d_3\\d_2+c_2d_3\\d_3 \end{pmatrix}
\right]. \]
Clearly, the vectors are a linear combination of $d_1,\,d_2,\,d_3$. Recalling the definition of $\mathbf{P}$ from the
Proposition, we get
\[\begin{pmatrix} d_1+c_1d_3\\d_2+c_2d_3\\d_3 \end{pmatrix}
=
\begin{pmatrix} 1 & 0 & c_1\\
0 & 1 & c_2\\
0 & 0 & 1 \end{pmatrix} \begin{pmatrix} d_1\\d_2\\d_3 \end{pmatrix}
=
 \mathbf{P}^{-1} \begin{pmatrix} d_1\\d_2\\d_3 \end{pmatrix} \]
  and simplify the exponent to
\[-\frac{1}{2}\left[
\begin{pmatrix} d_1\\d_2\\d_3 \end{pmatrix}
^{\text{T}} \left(\mathbf{P}^{-1}\right)^{\text{T}}
\Sigma^{-1}
\mathbf{P}^{-1}\begin{pmatrix} d_1\\d_2\\d_3 \end{pmatrix}
\right] \;.\]
Recognising that  $\left(\mathbf{P}^{-1}\right)^{\text{T}}
\Sigma^{-1}
\mathbf{P}^{-1} = \left(\mathbf{P}
\Sigma
\mathbf{P}^{\text{T}}\right)^{-1}$, we conclude the stated result.

Claim (2)(ii), on the Gaussian distribution for LPC if conditioned on $\beta$ and $\rho$, is trivial: simply recognise
that the terms $f_1(\beta,\rho)$ and $f_2(\beta,\rho)$ are fixed and can be thought of as modifications to the means of
$\dot\beta$ and $\dot\rho$.
\end{proof}

\subsection*{2. Details of~\cref{fig:vis_man,fig:vis_man_far}}
Each Figure considers a Gaussian prior in MPC / LPC coordinates using a diagonal $\Sigma$. The main challenge is to select
the mean and variance for $s$ and $\rho$ so that the priors are roughly comparable between the two coordinate systems.
In order to do so, we select an initial location for the target in Cartesian coordinates --- $(0.1,0.15)$~km
in~\cref{fig:vis_man} and $(1,1.5)$~km in~\cref{fig:vis_man_far} --- and an initial standard deviation for target location
uncertainty of $0.4$~km in Cartesian space. We sample 10,000 times from this Gaussian prior in Cartesian coordinates,
transform each sample into $s$ and $\rho$, and use the variance of the transformed samples as the variance of the prior
in LPC / MPC. The resulting standard deviations, and all other parameters for the figures, are shown in~\cref{tab:figs}.
The figures are then obtained by sampling 10,000 times from each prior and transforming those samples: label the $i^\text{th}$
sample $\bm{x}_i$, then the figures plot $\bm{x}_i^+ = \bm{h}(\bm{x}_i)$ for all $i$.

\begin{table}
\begin{center}
\begin{tabular}{ c c | c c c c c c c c }
 & & $\mu_\beta$ & $\mu_{\dot\beta}$ & $\mu_{\dot\rho}$ & $\mu_{s/\rho}$ & $\sigma^2_\beta$ & $\sigma^2_{\dot\beta}$ & $\sigma^2_{\dot\rho}$ & $\sigma^2_{s/\rho}$ \\
\begin{multirow}{2}{1.5em}{~\cref{fig:vis_man}
}\end{multirow}
& MPC & 0.58  & 1.1 & 0.1  & 5.5 & 0.05  & 0.2 & 0.2 & 4.0\\
& LPC  & 0.58  & 1.1 & 0.1  & -1.7 & 0.05  & 0.2 & 0.2 & 0.63 \\
 \hline
\begin{multirow}{2}{1.5em}{~\cref{fig:vis_man_far}
}\end{multirow}
& MPC & 0.58  & 1.1 & 0.1  & 0.55 & 0.05  & 0.2 & 0.2 & 0.14\\
& LPC  & 0.58  & 1.1 & 0.1  & 0.59 & 0.05  & 0.2 & 0.2 & 0.23
\end{tabular}
\end{center}
\caption{Means and variances for Gaussian priors, obtained as described in~\cref{app:vis_man}.}
\label{tab:figs}
\end{table}

\section{Demonstration of Higher-Order Moment Generation via Computer Algebra}
\label{app:ExampleMoment}
Our SageMath package, \href{https://github.com/athenax-acacia/seaweed}{\texttt{seaweed} \faGitSquare}, exports any
moment of the post-manoeuvre target state distribution as Python or LaTeX code. For example, we present the third raw
moment of bearing rate in LaTeX in~\cref{eq:ThirdRawBearingRate}. We have manually written the expected value on the
left-hand side, but we directly obtain the right-hand side by calling the function \texttt{moment\_0300/latex\_0300} in the
pre-generated \texttt{post\_manoeuvre} module.

There is a minor stylistic difference between the generated moments and the main body of this paper; the two components
of ownship speed change, $\Delta v_x$ and $\Delta v_y$, are respectively written as $\Delta v_0$ and $\Delta v_1$. We
have also manually added line breaks and alignment markers to fit the equation to the page.

\begin{figure*}[!t]
\begin{alignat}{3}
\label{eq:ThirdRawBearingRate}
\mathbb{E} \! \left[ x^3_1 \right] = &-\frac{3}{4} \, &&{\Delta v_0}^{3} \cos\left(\mu_{0} - 3 \,
\sigma_{30}\right) e^{\left(-3 \, \mu_{3} - \frac{1}{2} \, \sigma_{00} + \frac{9}{2} \, \sigma_{33}\right)} -
\frac{3}{4} \, {\Delta v_0} {\Delta v_1}^{2} \cos\left(\mu_{0} - 3 \, \sigma_{30}\right) e^{\left(-3 \, \mu_{3} -
\frac{1}{2} \, \sigma_{00} + \frac{9}{2} \, \sigma_{33}\right)} \nonumber \\
&- \frac{1}{4} \, &&{\Delta v_0}^{3} \cos\left(3 \, \mu_{0} - 9 \, \sigma_{30}\right) e^{\left(-3 \, \mu_{3} - \frac{9}{2}
\, \sigma_{00} + \frac{9}{2} \, \sigma_{33}\right)} + \frac{3}{4} \, {\Delta v_0} {\Delta v_1}^{2} \cos\left(3 \,
\mu_{0} - 9 \, \sigma_{30}\right) e^{\left(-3 \, \mu_{3} - \frac{9}{2} \, \sigma_{00} + \frac{9}{2} \,
\sigma_{33}\right)} \nonumber \\
&+ \frac{3}{4} \, &&{\Delta v_0}^{2} {\Delta v_1} e^{\left(-3 \, \mu_{3} - \frac{9}{2} \,
\sigma_{00} + \frac{9}{2} \, \sigma_{33}\right)} \sin\left(3 \, \mu_{0} - 9 \, \sigma_{30}\right) - \frac{1}{4} \,
{\Delta v_1}^{3} e^{\left(-3 \, \mu_{3} - \frac{9}{2} \, \sigma_{00} + \frac{9}{2} \, \sigma_{33}\right)} \sin\left(3
\, \mu_{0} - 9 \, \sigma_{30}\right) \nonumber \\
&+ \frac{3}{4} \, &&{\Delta v_0}^{2} {\Delta v_1} e^{\left(-3 \, \mu_{3} -
\frac{1}{2} \, \sigma_{00} + \frac{9}{2} \, \sigma_{33}\right)} \sin\left(\mu_{0} - 3 \, \sigma_{30}\right) +
\frac{3}{4} \, {\Delta v_1}^{3} e^{\left(-3 \, \mu_{3} - \frac{1}{2} \, \sigma_{00} + \frac{9}{2} \,
\sigma_{33}\right)} \sin\left(\mu_{0} - 3 \, \sigma_{30}\right) \nonumber \\
&+ \frac{3}{2} \, &&{\Delta v_0}^{2} {\left(\mu_{1} - 2
\, \sigma_{31}\right)} e^{\left(-2 \, \mu_{3} + 2 \, \sigma_{33}\right)} + \frac{3}{2} \, {\Delta v_1}^{2}
{\left(\mu_{1} - 2 \, \sigma_{31}\right)} e^{\left(-2 \, \mu_{3} + 2 \, \sigma_{33}\right)} \nonumber \\
&+ \frac{3}{2} \, &&{}
{\left({\left(\mu_{1} - 2 \, \sigma_{31}\right)} \cos\left(2 \, \mu_{0} - 4 \, \sigma_{30}\right) e^{\left(-2 \, \mu_{3}
- 2 \, \sigma_{00} + 2 \, \sigma_{33}\right)} - 2 \, \sigma_{10} e^{\left(-2 \, \mu_{3} - 2 \, \sigma_{00} + 2 \,
\sigma_{33}\right)} \sin\left(2 \, \mu_{0} - 4 \, \sigma_{30}\right)\right)} {\Delta v_0}^{2} \nonumber \\
&- 3 \, &&{} {\left(2 \,
\sigma_{10} \cos\left(2 \, \mu_{0} - 4 \, \sigma_{30}\right) e^{\left(-2 \, \mu_{3} - 2 \, \sigma_{00} + 2 \,
\sigma_{33}\right)} + {\left(\mu_{1} - 2 \, \sigma_{31}\right)} e^{\left(-2 \, \mu_{3} - 2 \, \sigma_{00} + 2 \,
\sigma_{33}\right)} \sin\left(2 \, \mu_{0} - 4 \, \sigma_{30}\right)\right)} {\Delta v_0} {\Delta v_1} \nonumber \\
&- \frac{3}{2}
\, &&{} {\left({\left(\mu_{1} - 2 \, \sigma_{31}\right)} \cos\left(2 \, \mu_{0} - 4 \, \sigma_{30}\right) e^{\left(-2 \,
\mu_{3} - 2 \, \sigma_{00} + 2 \, \sigma_{33}\right)} - 2 \, \sigma_{10} e^{\left(-2 \, \mu_{3} - 2 \, \sigma_{00} + 2
\, \sigma_{33}\right)} \sin\left(2 \, \mu_{0} - 4 \, \sigma_{30}\right)\right)} {\Delta v_1}^{2} \nonumber \\
&- 3 \, &&{}
{\left(\sigma_{11} \cos\left(\mu_{0} - \sigma_{30}\right) e^{\left(-\mu_{3} - \frac{1}{2} \, \sigma_{00} + \frac{1}{2}
\, \sigma_{33}\right)} + {\left({\left(\mu_{1} - \sigma_{31}\right)} \cos\left(\mu_{0} - \sigma_{30}\right)
e^{\left(-\mu_{3} - \frac{1}{2} \, \sigma_{00} + \frac{1}{2} \, \sigma_{33}\right)} \right.} \right.} \nonumber \\
&{} &&- {\left. {\left.\sigma_{10} e^{\left(-\mu_{3} -
\frac{1}{2} \, \sigma_{00} + \frac{1}{2} \, \sigma_{33}\right)} \sin\left(\mu_{0} - \sigma_{30}\right)\right)}
{\left(\mu_{1} + \sigma_{10} + \sigma_{31}\right)}\right)} {\Delta v_0} \nonumber \\
&+ 3 \, &&{} {\left(\sigma_{11} e^{\left(-\mu_{3} -
\frac{1}{2} \, \sigma_{00} + \frac{1}{2} \, \sigma_{33}\right)} \sin\left(\mu_{0} - \sigma_{30}\right) +
{\left(\sigma_{10} \cos\left(\mu_{0} - \sigma_{30}\right) e^{\left(-\mu_{3} - \frac{1}{2} \, \sigma_{00} + \frac{1}{2}
\, \sigma_{33}\right)} \right.} \right.} \nonumber \\
&{} &&+ {\left. {\left. {\left(\mu_{1} - \sigma_{31}\right)} e^{\left(-\mu_{3} - \frac{1}{2} \, \sigma_{00} +
\frac{1}{2} \, \sigma_{33}\right)} \sin\left(\mu_{0} - \sigma_{30}\right)\right)} {\left(\mu_{1} - \sigma_{10} +
\sigma_{31}\right)}\right)} {\Delta v_1} + {\left(\mu_{1}^{2} + \sigma_{11}\right)} \mu_{1} + 2 \, \mu_{1} \sigma_{11} \text{.}
\end{alignat}
\end{figure*}
\end{appendices}

\bibliography{IEEEabrv,ukf.bib}

\begin{thebibliography}{10}

\bibitem{Aidala:1983:Utilization}
V.~Aidala and S.~Hammel.
\newblock Utilization of modified polar coordinates for bearings-only tracking.
\newblock {\em {IEEE} Transactions on Automatic Control}, 28(3):283--294, March 1983.

\bibitem{Hoelzer:1978:Modified}
Hans~D. Hoelzer, G.~W. Johnson, and A.~O. Cohen.
\newblock {Modified Polar Coordinates - The Key to Well Behaved Bearings Only Ranging}.
\newblock Technical Report IBM IR\&D Report 78-M19-0001A, IBM Federal Systems Division, 1978.

\bibitem{Wang:2010:Algorithm}
Dan Wang, Hongyan Hua, and Haiwang Cao.
\newblock Algorithm of modified polar coordinates {UKF} for bearings-only target tracking.
\newblock In {\em {2010 2nd International Conference on Future Computer and Communication}}, volume~3, pages V3--557--V3--560, 2010.

\bibitem{Brehard:2006:Closed}
T.~Brehard and J.~R.~Le Cadre.
\newblock Closed-form posterior {C}ramér-{R}ao bounds for bearings-only tracking.
\newblock {\em {IEEE} Transactions on Aerospace and Electronic Systems}, 42(4):1198--1223, October 2006.

\bibitem{Mallick:2011:Angle}
Mahendra Mallick, Sanjeev Arulampalam, Lyudmila Mihaylova, and Yanjun Yan.
\newblock Angle-only filtering in {3D} using {M}odified {S}pherical and {L}og {S}pherical {C}oordinates.
\newblock In {\em Proceedings of the 14th International Conference on Information Fusion ({FUSION})}, pages 1--8, Chicago, IL, July 2011. IEEE.

\bibitem{Haykin:2001:Kalman}
Simon Haykin, editor.
\newblock {\em {Kalman Filtering and Neural Networks}}.
\newblock John Wiley {\&} Sons, Inc., October 2001.

\bibitem{Wan:2000:Unscented}
E.~A. Wan and R.~van~der Merwe.
\newblock {The unscented Kalman filter for nonlinear estimation}.
\newblock In {\em Proceedings of the {IEEE} 2000 Adaptive Systems for Signal Processing, Communications, and Control Symposium (Cat. No.00EX373)}. {IEEE}, 2000.

\bibitem{Houtekamer:1998:EnKF}
P.~L. Houtekamer and Herschel~L. Mitchell.
\newblock Data assimilation using an ensemble kalman filter technique.
\newblock {\em Monthly Weather Review}, 126(3):796--811, March 1998.

\bibitem{Schmidt:2012:Stochastics}
Volker Schmidt.
\newblock Stochastics {III} lecture notes.
\newblock Ulm University, 2012.

\bibitem{Bryc:1995:Normal}
Wlodzimierz Bryc.
\newblock {\em The {N}ormal {D}istribution: {C}haracterizations with {A}pplications}.
\newblock Springer-{V}erlag, 1995.

\bibitem{TheSageDevelopers:2023:SageMath}
{The Sage Developers}, William Stein, David Joyner, David Kohel, John Cremona, and Burçin Eröcal.
\newblock {\em {S}ageMath, the {S}age {M}athematics {S}oftware {S}ystem ({V}ersion 9.8)}, 2023.
\newblock {\tt https://www.sagemath.org}.

\bibitem{Labbe:2014:Kalman}
Roger Labbe.
\newblock Kalman and {B}ayesian {F}ilters in {P}ython.
\newblock \url{https://github.com/rlabbe/Kalman-and-Bayesian-Filters-in-Python}, 2014.

\bibitem{Bar-Shalom:2001:Estimation}
Yaakov Bar-Shalom, X.~Rong Li, and Thiagalingam Kirubarajan.
\newblock {\em Estimation with {A}pplications to {T}racking and {N}avigation}.
\newblock John {W}iley and {S}ons, 2001.

\bibitem{Zanetti:2018:Novel}
Renato Zanetti and Kirsten Tuggle.
\newblock A novel {G}aussian {M}ixture approximation for nonlinear estimation.
\newblock In {\em 2018 21st International Conference on Information Fusion (FUSION)}. IEEE, July 2018.

\end{thebibliography}
\bibliographystyle{unsrt}

\section*{Biography Section}
\vskip -2\baselineskip plus -1fil
\begin{IEEEbiographynophoto}{Athena Helena Xiourouppa}
Having completed a Bachelor of Mathematical Sciences (Advanced) at the University of Adelaide, Athena is pursuing
postgraduate research in Statistics and Applied Mathematics via a Master of Philosophy. Her studies are supported by
Acacia Systems. Athena has proficiency in statistical modelling, simulation, and dynamical
systems.

Outside of her work and study, Athena holds a position on the Women in Science, Technology, Engineering, and Mathematics
Society (WISTEMS) committee at the University of Adelaide. This allows her to meaningfully apply her experience in
academia and industry to inspire women and other minorities in STEM.
\end{IEEEbiographynophoto}
\vskip -2\baselineskip plus -1fil
\begin{IEEEbiographynophoto}{Dmitry Mikhin} Dmitry started his research career in early 90' working on sound propagation
    modelling in the ocean using ray acoustics, and later, parabolic equation methods. Since then he worked in areas as
    varied as acoustical modelling, infrared propagation in the atmosphere, flight planning and route search, asset
    optimisation and planning in mining, agricultural robotics, and lately, estimation, tracking, and data fusion.

    When wearing his second hat, Dmitry is a professional software developer combining his research and coding skills to
    bridge the vast badlands separating the academic science and industry applications.
\end{IEEEbiographynophoto}
\vskip -2\baselineskip plus -1fil
\begin{IEEEbiographynophoto}{Melissa Humphries}
Dr. Melissa Humphries is a statistician --- she creates, and applies, analytical tools that make sense of our world.
Working in areas like forensic science, health and psychology, the goal is always to increase efficiency, accuracy, and
transparency. Building bridges between machines and experts, Melissa's work aims to support experts in making decisions
in an explainable way.

Working at the interface between statistics and AI, Melissa's research draws meaning out of complex systems by
leveraging the explainable components of statistical and physical models and elevating them using novel techniques. She
strongly advocates the integration of the end-user from the outset of all research and is committed to finding impactful
ways to implement science.

Melissa is also a passionate advocate of equity across all areas of her work and life.
\end{IEEEbiographynophoto}
\vskip -2\baselineskip plus -1fil
\begin{IEEEbiographynophoto}{John Maclean}
I use Bayesian methods to attack inference problems in which one has access to forecasts from a noisy, incorrect,
chaotic model, and also access to noisy, incomplete data. The key question is how to use forecasts and data together
in order to infer model states or parameters. The methods I use work very well in low dimensional problems, and my
research attacks the curse of dimensionality found when employing high dimensional data.
\end{IEEEbiographynophoto}

\end{document}